\documentclass[journal]{IEEEtran}

\usepackage{comment}
\usepackage[section]{placeins}
\usepackage{url}  
\usepackage{graphicx}  
\usepackage{color}
\usepackage{comment}
\usepackage{hyperref}
\usepackage{array}
\usepackage{graphicx,subfigure}
\usepackage[noend]{algpseudocode}
\usepackage{algorithm}

\usepackage{amsmath, amsthm, amssymb}

\usepackage{mdwlist}
\usepackage[short]{optional}

\usepackage[utf8]{inputenc}
\usepackage{CJK}
\usepackage[left=1.5cm, right=1.5cm, top=1.5cm, bottom=1.5cm]{geometry}

\usepackage{enumitem}
\usepackage{threeparttable}
\usepackage{booktabs}
\usepackage{diagbox}
\usepackage{stfloats}
\newtheorem{theorem}{Theorem}

\newtheorem{definition}{Definition}

\algnewcommand\And{\textbf{and}}

\newcommand{\shenr}[1] 
{
\textbf{\color{blue}{Shen: #1}}
}
\newcommand{\shenc}[1] 
{
\textbf{\color{red}{Shen: #1}}
}
\newcommand{\yehc}[1] 
{
\textbf{\color{red}{Shen: #1}}
}

\newcommand{\tsengc}[1] 
{
\textbf{\color{magenta}{Tseng: #1}}
}

\newcommand{\phead}[1]
{
\noindent \textbf{#1}
}

\begin{document}
\begin{CJK}{UTF8}{bkai}

\title{Auditable Homomorphic-based Decentralized Collaborative AI with Attribute-based Differential Privacy
}

\author{Lo-Yao Yeh, Sheng-Po Tseng, Chia-Hsun Lu, and Chih-Ya Shen

\IEEEcompsocitemizethanks{
\IEEEcompsocthanksitem L.-Y. Yeh is with the Department of Information Management, National Central University, Taoyuan 320, Taiwan (e-mail: yehloyao@ncu.edu.tw).%
\IEEEcompsocthanksitem S.-P. Tseng, C.-H. Lu, C.-Y. Shen are with the Department of Computer Science, National Tsing Hua University, Hsinchu, Taiwan (e-mail: chihya@cs.nthu.edu.tw).}}

\maketitle

\begin{abstract}
\normalsize
In recent years, the notion of federated learning (FL) has led to the new paradigm of distributed artificial intelligence (AI) with privacy preservation. However, most current FL systems suffer from data privacy issues due to the requirement of a trusted third party. Although some previous works introduce differential privacy to protect the data, however, it may also significantly deteriorate the model performance. To address these issues, we propose a novel decentralized collaborative AI framework, named \emph{\underline{A}uditable Homomorphic-based D\underline{e}cent\underline{r}al\underline{is}ed Collaborative AI (AerisAI)}, to improve security with homomorphic encryption and fine-grained differential privacy. Our proposed AerisAI directly aggregates the encrypted parameters with a blockchain-based smart contract to get rid of the need of a trusted third party. We also propose a brand-new concept for eliminating the negative impacts of differential privacy for model performance. Moreover, the proposed AerisAI also provides the broadcast-aware group key management based on ciphertext-policy attribute-based encryption (CP-ABE) to achieve fine-grained access control based on different service-level agreements. We provide a formal theoretical analysis of the proposed AerisAI as well as the functionality comparison with the other baselines. We also conduct extensive experiments on real datasets to evaluate the proposed approach. The experimental results indicate that our proposed AerisAI significantly outperforms the other state-of-the-art baselines. 
\end{abstract}


\section{Introduction} 
\label{c:Introduction}

With the emergence of deep learning, an increasing number of companies and organizations train deep learning models with their private data for a wide spectrum of applications, such as medical diagnosis, commercial analysis, recommendation for e-commerce, and surveillance systems. To train an accurate model, it is necessary to collect a large amount of high-quality data. However, not every organization has a sufficient amount of data for training an accurate model in the real world. To solve this problem, sharing data with each other would be a straightforward approach. However, an issue arises -- data privacy. Due to privacy concerns, organizations may hesitate to exchange/share data with others. This is also the main obstacle in applying large-scale deep learning models in real-world scenarios.

To tackle the privacy issue, the notion of \emph{Federated Learning (FL)}\footnote{In this paper, we focus on the most widely adopted \emph{Horizontal Federated Learning (HFL)}. For the ease of presentation, we use the terms \emph{Horizontal Federated Learning (HFL)} and \emph{Federated Learning (FL)} interchangeably in this paper.}~\cite{mcmahan2017communication} came out for large-scale machine learning. Instead of sharing/exchanging the private data among the participating organizations directly, in Horizontal Federated Learning (HFL), these organizations collaboratively train a global deep learning model without disclosing their private data. Conceptually, each client of the HFL is connected to a centralized server. Each client first trains the local model with its private data and then sends the gradients to the centralized server after local training. The server is responsible for updating the global model by aggregating the gradients from all the clients. Afterward, each client replaces the local model weights with the updated global model downloaded from the server.

Federated learning has been employed in many application scenarios. For example, FL is applied on the Internet of Things (IoT)~\cite{zhang2020efficient, baghban2022edge, leng2023manuchain, wang2023eidls} to enable collaboration among a massive number of edge devices. FL is also applied to healthcare~\cite{chen2020fedhealth, yan2020variation}. The global model is trained with the data of user activities to achieve accurate personal healthcare with FL. Moreover, FL is introduced to address the privacy issues in the intelligent transportation systems \cite{zhao2022decentralized}.


However, some weaknesses remain in the widely-adopted FL frameworks \cite{mcmahan2017communication, aono2017privacy,Nagalapatti_Mittal_Narayanam_2022}, as listed below.  

\phead{W1.} \emph{Requiring a centralized server}. The centralized server plays a critical role in FL to aggregate the gradients and store the global model. However, there is a potential risk associated with this single point of failure. If the server fails or is compromised, the entire system collapses.

\phead{W2.} \emph{Lack of auditability}. Conventional FL works in a less transparent fashion, and only the centralized server keeps track of the uploaded statistics from the clients. This jeopardizes the trust of the clients, because they do not have any information to judge whether any \emph{lazy} client exists (i.e., a client that does not contribute their private data to train its local model and upload the gradients), unless the centralized server explicitly discloses such information.

\phead{W3.} \emph{Privacy issues -- gradient leakage and non-differential privacy}. Although FL enables privacy-preserving distributed machine learning by sending only the gradients (instead of the private data) from the clients to the server, however, the uploaded gradients may still leak the clients' private information~\cite{zhu2019deep, geiping2020inverting,huang2009provable}. In addition, conventional FL does not adopt the differential privacy technique. As a consequence, the gradients may be inferred by comparing the global models before and after the model updates. For instance, assume that the original global model parameters are $[0.05, -0.04]$. After an update, the model parameters become $[0.02, -0.01]$. In this case, one can infer that the gradients are $[-0.03, 0.03]$ by subtracting the original parameters from the updated ones.



To address the weakness W1 above (i.e., requiring a centralized server), some previous works integrated blockchain technology with FL~\cite{chen2021ds2pm, guo2022sandbox, xu2022spdl, kalapaaking2022blockchain}. Decentralization is one of the characteristics of blockchain to avoid the need for a centralized server. 
However, there is a major security risk in previous works~\cite{chen2021ds2pm, kalapaaking2022blockchain}, i.e., when aggregating and updating the gradients to the global model, the gradients may leak, even if encryption schemes are applied. This is because previous works~\cite{chen2021ds2pm, kalapaaking2022blockchain} need to decrypt the gradients before the aggregation. Unfortunately, Guo et al.~\cite{guo2022sandbox} do not protect the model information in the whole procedure. In addition, several works \cite{xu2022spdl, wang2021protecting, truex2019hybrid, li2020privacy} perturb the gradients to mask the real values based on differential privacy \cite{dwork2006calibrating}, but the model performance significantly deteriorates when the privacy budget is decreased (i.e., the data privacy protection is enhanced).
In summary, previous approaches that integrate blockchain with FL either do not fully address the privacy issue or result in serious performance degradation.

\sloppy
To tackle these critical issues, in this paper, we propose \emph{\underline{A}uditable Homomorphic-based D\underline{e}cent\underline{r}al\underline{is}ed AI (AerisAI)}, which achieves excellent model performance while preserving the privacy of the clients. The proposed \emph{AerisAI} addresses all the above weaknesses as below.

\begin{itemize}
    \item \emph{Decentralization.} To address the weakness W1, we employ blockchain to build a decentralized platform. In this case, the weaknesses of requiring a centralized server, e.g., a single point of failure and performance bottlenecks, no longer exists.
    
    \item \emph{Auditability.} All the transactions on the blockchain are recorded and can be audited by all the clients. This addresses the weakness W2. The full transparency boosts the clients' confidence in the system.
    
    \item \emph{Privacy preservation.} To address the weakness W3, we perturb the gradients by adding noise to prevent the real values of the gradients from being exposed on blockchain. We also use homomorphic encryption to encrypt the perturbed gradients and noise. Besides, we devise the aggregation algorithm in the smart contract that automatically aggregates the encrypted perturbed gradients and noise to update the global model. Notice that we simultaneously achieve both gradients perturbation and masking, which has not been provided in current works. 
    
    \item \emph{Group key management.} To address the degradation of global model performance caused by noise injection in federated learning, we propose a method wherein clients can encrypt the noise and upload it through a smart contract for aggregation. The blockchain oracle is responsible for distributing the aggregated noise, which enables clients to mitigate the negative impact of the noise. In traditional approaches, the oracle encrypts the aggregated noise using a specific client's public key and then distributes it to the corresponding client, who can decrypt the encrypted aggregated noise with its private key. However, this method becomes time-consuming when dealing with multiple clients participating in model training. With an increasing number of clients, the oracle takes more time to distribute the aggregated noise to each client, which poses scalability problems. To address the issue of scalability, we propose the use of group key management based on ciphertext policy attribute-based encryption (CP-ABE) \cite{bethencourt2007ciphertext}, allowing the oracle to broadcast the aggregated noise to clients. AerisAI is the first approach to integrate this proposed group key management based on the CP-ABE cryptosystem into federated learning, enabling the blockchain oracle to distribute the encrypted aggregated noise more efficiently. 
\end{itemize}

In this paper, we design the protocol and workflow of AerisAI, and we formally prove its security. We provide a functionality comparison to demonstrate that the proposed AerisAI addresses all the privacy and security issues while the other state-of-the-art baselines do not. We also conduct extensive experiments on real machine learning datasets to demonstrate the proposed AerisAI outperforms the other baselines significantly. The contributions of this paper are summarized as follows.

\begin{itemize}
    \item We introduce blockchain to enable a decentralized and transparent platform. The clients do not need a centralized server to store and update the global model. Instead, the smart contract automatically completes the operations the centralized server offers.
    
    \item To our best knowledge, the proposed AriesAI is the first decentralized collaborative AI framework that combines gradients perturbation, gradients masking, and group key management. The gradients perturbation and masking effectively improve the security, while group key management significantly improves the scalability of the blockchain-enabled platform.
    
    \item We formally prove the security of the proposed AerisAI and compare it with other state-of-the-art approaches in terms of functionality. 
    
    \item We perform extensive experiments on real benchmark dataset and compare with multiple state-of-the-art baselines to demonstrate the accuracy of the global model. The results indicate that the performance of the global model is not degraded by our proposed techniques that significantly improve the security strength.
\end{itemize}

The remainder of this paper is organized as follows. Sec. \ref{c:Relatedwork} introduces the related work. Sec. \ref{c:System} discusses the system model and assumptions. The proposed framework is detailed in Sec. \ref{c:Scheme}. In Sec. \ref{c:Securityanalysis}, we present the formal security analysis. Sec. \ref{c:Evaluatoin} presents the functionality comparisons and the experimental results. Finally, Sec. \ref{c:Conclusions} concludes this paper.


\section{Related Works}
\label{c:Relatedwork}
\phead{Federated Learning.}
The concept of Federated Learning (FL) is to allow multiple clients collaboratively train a machine learning model with a centralized server, while keeping the training data of each client decentralized and private~\cite{yang2019federated}. The private data from each client are not shared with the other clients or centralized server, which mitigates the privacy risks of traditional machine learning that needs to upload and combine the private data from each client before training. 

McMahan et al. \cite{mcmahan2017communication} first propose the FL framework for mobile devices. Each client initializes its local model by first downloading the global model and trains the local model using its private dataset. Afterward, the server randomly samples several clients and computes the weighted average of the model parameters, which are then employed as the global model in the next round. Phong et al. \cite{aono2017privacy} introduce additively homomorphic encryption for the FL scheme. The clients encrypt gradients before uploading them to the server, which avoids information leakage to the server. Due to the characteristics of additively homomorphic encryption, the server is thus able to correctly aggregate the ciphertext and update the encrypted global model.
Naglapatti et al. \cite{Nagalapatti_Mittal_Narayanam_2022} build an effective global model by training a policy network based on reinforcement learning (RL)~\cite{sutton1999policy}. The RL agent helps the clients select the relevant data samples that are benign to the global model from the private dataset for local training. 

Although the above FL frameworks achieve good performance while keeping the private data decentralized~\cite{mcmahan2017communication,aono2017privacy,Nagalapatti_Mittal_Narayanam_2022,wang2023eidls,zhao2022decentralized}, they still rely on a centralized server to aggregate gradients from the clients. This centralized server may lead to potential security risks, i.e., a single point of failure. In our proposed AerisAI, we employ blockchain~\cite{nakamoto2008bitcoin} to replace the centralized server. In this way, the distributed nature of blockchain solves the weakness of the centralized server and alleviates the burden of server maintenance.

\phead{Blockchain-Based Collaborative AI Platform.}
In 2008, Nakamoto first introduces the concept of blockchain, which is a public ledger that records all the transactions on the list of blocks~\cite{nakamoto2008bitcoin}. Blockchain-related techniques have also been applied to FL. Chen et al. \cite{chen2021ds2pm} construct a platform of FL that combines blockchain, in which a data requester invokes a transaction, and the data providers whose private data are similar to the requested data collaboratively train a global model. After the providers train the global model, they encrypt the model with homomorphic encryption and send it to the requester. 

The concepts of the sandbox and state channel are introduced to mobile edge computing systems that combine blockchain and FL~\cite{guo2022sandbox}. The state channel is used to create a trusted sandbox that aggregates the local model into the global model and offers suitable rewards to each device. Xu et al. \cite{xu2022spdl} introduce an aggregation rule, called Krum~\cite{blanchard2017machine}, to select the benign gradients for updating the global model. This rule filters Byzantine gradients at each round to maintain the performance of the global model. Besides, each client adds noise to the gradients before uploading them to the blockchain, to avoid leaking the private data from gradients~\cite{zhu2019deep,geiping2020inverting}.

Kalapaaking et al. \cite{kalapaaking2022blockchain} propose a blockchain-based AI framework with Intel Software Guard Extension (SGX)~\cite{costan2016intel} in industrial Internet-of-Things. The SGX provides an isolated environment in which the edge server aggregates the local models. Since previous research indicate that there exists a chance that private data are recovered from the gradients or model parameters~\cite{zhu2019deep,geiping2020inverting}, the values of the parameters are not directly exposed during the process. However, the model parameters are not masked before uploading to the edge server and blockchain~\cite{guo2022sandbox,kalapaaking2022blockchain}, which may lead to significant parameters/data leakage. 
 
In the above approaches~\cite{chen2021ds2pm,kalapaaking2022blockchain}, the encrypted local models need to be decrypted before aggregation, leading to a risk of private data leakage. To tackle this important privacy issue, in this paper, we encrypt the global model and gradients by homomorphic encryption, which are not decrypted during the aggregation. Another potential approach to prevent private data leakage is to employ the idea of differential privacy~\cite{xu2022spdl} that perturbs the gradients with noise. However, adding noise degrades the performance of the global model. In this paper, we carefully design the mechanism, named \emph{Attribute-Based Differential Privacy}, to avoid the performance drop of the global model while preserving data privacy.

\phead{Application scenarios for the proposed approach.} The proposed approach has a wide spectrum for applications, such as the communication and group organization in social networks~\cite{shen2015socio,shen2017finding,shen2020activity,shen2020extracting}, e-commerce~\cite{chang2019blockchain,chen2018efficient,lin2019smart}, collaborative learning and sensing~\cite{yang2021learning,chang2019she,yeh2022gdpr}, machine learning applications that requires enhanced privacy preservation~\cite{chang2022learning,yang2022enhancing}.

\section{System Model and Assumptions}
\label{c:System}
\phead{System Model.}
The system proposed in this paper is composed of the following major components. 

\emph{Clients.} Each client has a machine learning model (local model) and its private data that are not shared with others or the server. Each client employs its training data to train its local model and obtains the gradients from the training process. The generated gradients that are used to update the global model are encrypted and uploaded to the blockchain. Following previous assumptions adopted by most FL approaches, the clients are assumed to be honest-but-curious~\cite{nagalapatti2021game, Nagalapatti_Mittal_Narayanam_2022, mohassel2017secureml}, meaning they will follow the protocol in the system but will attempt to learn the other clients' private information as much as possible. In other words, each client truthfully performs the designed operations to produce the correct outcome, but each client would also attempt to learn other clients' information from the uploaded gradients.

\emph{Blockchain.} As a form of public ledger, blockchain records each transaction on the chain. In this paper, we leverage Hyperledger Fabric to develop our blockchain network for recording the uploaded gradients from the clients. Note that the encrypted global model is also stored on the blockchain. This blockchain is a consortium blockchain that relies on a Byzantine fault-tolerant (BFT) protocol to reach a consensus. Typically, it creates a secure environment by assuming no more than $\frac{1}{3}$ of malicious clients in the system~\cite{androulaki2018hyperledger}. 

\emph{Smart contract.} The smart contract acts as a bridge between the clients and the blockchain. 
The clients upload gradients and download the global model from the blockchain via the smart contract. Therefore, the aggregation method is defined in the smart contract, and the smart contract automatically aggregates the uploaded gradients according to the definition.
    
\emph{Oracle.} In our scenario, the blockchain oracle~\cite{xu2016blockchain, moudoud2019iot} is introduced to decrypt the encrypted noise uploaded by the clients. Note that the oracle without the corresponding private key cannot decrypt the global model, and thus the privacy of the clients is guaranteed.

\phead{Security Model.} As the blockchain is employed, the transactions recorded on the ledger are hard to tamper with and transparent to everyone (including all the clients). The clients are honest-but-curious, so they may dig for other clients' data or other information from the encrypted gradients recorded on the ledger. The clients may also collaborate to break data privacy. Also, potential adversaries outside the system will try to recover the clients' private data from the encrypted gradients. Similar to previous works, it is assumed that the clients and adversaries are polynomial-time entities~\cite{yeh2022blockchain, zhao2022practical}.

\phead{Security Goal.} The security goal of our proposed AerisAI is to address the weaknesses identified in Sec. \ref{c:Introduction}: \emph{W1. requiring a centralized server}, \emph{W2. lack of auditability}, and \emph{W3. privacy issues}.

\section{The Proposed Scheme -- \emph{AerisAI}}
\label{c:Scheme}

We first briefly introduce the preliminaries and then detail the workflow.

\subsection{Preliminaries}
\subsubsection{Bilinear Pairing}
Let $G_1$ and $G_T$ be two multiplicative cyclic groups of prime order $p$. Let $g$ be the generator of $G_1$ and $e$ be a bilinear map, $e:G_1 \times G_1\to G_T$, which satisfies the following properties.
\begin{enumerate}
    \item [(1)] Bilinearity: For all $u, v\in G_1$ and any $a, b\in \mathbb{Z}^*_p$, the equation $e(u^a, v^b) = e(u, v)^{ab}$ holds. 
    \item [(2)] Non-degeneracy: $e(g, g) \neq 1$.
    \item [(3)] Computability: For $u, v\in G_1$, there exists an efficient algorithm to compute $e(u, v) \in G_T$.
\end{enumerate}

\subsubsection{Ciphertext Policy Attribute-Based Encryption (CP-ABE)}
Attribute-based encryption (ABE) \cite{sahai2005fuzzy} is a type of public-key cryptography system but it is not necessary to generate a pair of keys which consistof a public key and a private key. Instead, the ciphertext is encrypted with associated attributes, and the private key is produced according to the user's attributes. In this paper, we focus on ciphertext policy attribute-based encryption (CP-ABE) \cite{bethencourt2007ciphertext}, which is an advanced application of ABE. The concept is that encrypting the data with the specific access policy over a set of attributes, and only the user whose attributes comply with the policy is able to decrypt the ciphertext. This encryption algorithm can be divided into four parts: \emph{setup}, \emph{key generation}, \emph{encryption}, and \emph{decryption}.

\begin{itemize}[leftmargin=*]
    \item \textbf{\emph{Setup$(1^\lambda)\to (PK, MK)$}:}  In the beginning, a bilinear map $e$ and a bilinear group $G_1$ of prime order $p$ with generator $g$ are  defined according to the input regarding the security parameter $(1^\lambda)$. Then, the algorithm randomly chooses two numbers $\alpha,\beta\in\mathbb{Z}^*_p$. The public key ($PK$) and master key ($MK$) are generated as:
    \[
        PK=(g, Q=g^\beta, e(g,g)^\alpha), MK=(\beta, g^\alpha).
    \]
    \item \textbf{\emph{KeyGeneration$(MK, \mathcal{A})\to (PrivKey)$}:} In this step, the algorithm generates the set of private keys according to the associated attributes $\mathcal{A}$. It chooses a random number $r\in\mathbb{Z}^*_p$. For each attribute $i\in\mathcal{A}$, it is also assigned a random number $r_i\in\mathbb{Z}^*_p$ and takes as an input in a hash function $H:\{0,1\}^*\to G_1$, which outputs the corresponding hash value. Then, the private key ($PrivKey$) is generated as:
    \[
        PrivKey=(S=g^{\frac{\alpha+r}{\beta}},
        \forall{i}\in\mathcal{A}: S_i=g^rH(i)^{r_i}, S^{'}_{i}=g^{r_i}).
    \]
    \item \textbf{\emph{Encryption$(PK, SEK, \mathcal{P})\to (CT_{SEK})$}:}  In our proposed AerisAI, the session key $SEK$ is encrypted under the policy $\mathcal{P}$, which is a tree access structure. We define each node as $x$. Each non-leaf node represents a threshold gate whose threshold value is $q_x$. Then, the algorithm randomly chooses a polynomial $l_x$ that the degree is set as $d_x=q_x-1$. A random secret value $s\in\mathbb{Z}^*_p$ is set as $l_R(0)=s$ for the root $R$. Therefore, the setting of the value of each other node is based on the polynomial of its parent node and the index of the node, i.e., $l_x(0)=l_{parent(x)}(index(x))$. We note that the polynomial coefficients whose degrees are not zero are randomly chosen. The ciphertext ($CT$) is finally computed as below. Let $F$ denote the set of the leaf nodes in $\mathcal{P}$.
    \[
    \begin{split}
        CT=&(\mathcal{P}, \Tilde{C}=SEKe(g,g)^{\alpha s}, C=Q^s, \\
        &\forall{f}\in{F}: C_f=g^{l_f(0)}, C^{'}_{f}=H(attribute(f))^{l_f(0)}.
    \end{split}
    \]
    \item \textbf{\emph{Decryption$(CT, PrivKey)\to (SEK)$}:} The decryption method is a recursive algorithm. If the user's attributes satisfy the policy, the ciphertext can be decrypted into plaintext. Let node $f$ be a leaf node and $j$ be the attributes of $f$. If $j\in\mathcal{A}$, the decryption procedure is as follows:
    \[
    \begin{split}
        &\emph{NodeDecryption$(CT, PrivKey, f)$} \\
        &= \frac{e(S_j, C_f)}{e(S^{'}_{j}, C^{'}_f)} =\frac{e(g^{r}H(j)^{r_j}, g^{l_f(0)})}{e(g^{r_j}, H(j)^{l_f(0)})}\\
        &=\frac{e(g^r, g^{l_f(0)})e(H(j)^{r_j}, g^{l_f(0)})}{e(g^{r_j}, H(j)^{l_f(0)})}\\
        &=e(g^r, g^{l_f(0)})=e(g,g)^{r l_f(0)}
    \end{split}
    \] 
    In the case that node $x$ is a non-leaf node, the procedure should necessarily consider the children nodes $y$ of $x$. It recursively calls \emph{NodeDecryption$(CT, PrivKey, x)$} and stores the output $D_y$ of each child node $y$. $Y_x$ is defined as a $q_x$-size set of $y$ such that $D_y\neq \perp $ and $D_x$ is computed as:
    \[
    \begin{split}
        D_x&=\prod_{y\in{Y_x}}D_y^{\Delta_{k, {Y^{'}_{x}}^{(0)}}}
        (k=index(y), {Y^{'}_{x}}={index(y):y\in Y_x})\\
        &=\prod_{y\in{Y_x}}(e(g,g)^{r {l_{y}(0)}})^{\Delta_{k, {Y^{'}_{x}}^{(0)}}} \\
        &=\prod_{y\in{Y_x}}(e(g,g)^{r {l_{parent(y)}(index(y))}})^{\Delta_{k, {Y^{'}_{x}}^{(0)}}} \\
        &=\prod_{y\in{Y_x}}(e(g,g)^{r {l_{x}(k)}})^{\Delta_{k, {Y^{'}_{x}}^{(0)}}}=e(g,g)^{r l_x(0)} 
    \end{split}
    \]
    
    Here, $\Delta_{k, {Y^{'}_{x}}^{(0)}}$ is the Lagrange coefficient \cite{bethencourt2007ciphertext}. The algorithm is initially implemented on the root of $\mathcal{P}$. If all the attributes satisfy the policy, we can retrieve the secret value $s$ of the root from $l_R(0)$. Finally, the session key is 
    obtained by the user.
    \[
    \begin{split}
        \frac{\Tilde{C}}{\frac{e(C, S)}{e(g,g)^{r l_R(0)}}}
        &= \frac{\Tilde{C}}{\frac{e(g^{\beta s}, Q^{\frac{\alpha+r}{\beta}})}{e(g,g)^{rs}}}=\frac{\Tilde{C}}{\frac{e(g,g)^{s(\alpha+r)}}{e(g,g)^{rs}}}\\
        &=\frac{\Tilde{C}}{e(g,g)^{\alpha s}}=\frac{SEK e(g,g)^{\alpha s}}{e(g,g)^{\alpha s}}=SEK
    \end{split}
    \]
\end{itemize}

\subsubsection{Paillier Encryption}
Homomorphic encryption (HE) \cite{acar2018survey} allows users to compute the ciphertext without decrypting them first. After decryption, the result of the computation on the ciphertext is the same as on the plaintext. Two common types of HE include \emph{Fully Homomorphic Encryption (FHE)} and \emph{Partially Homomorphic Encryption (PHE)}. The former supports arbitrary computation, while the latter permits only the particular computation on the ciphertext. In our proposed AerisAI, we focus on PHE, which is used to encrypt the gradients to defer malicious people from obtaining the true values while allowing legitimate computation on the encrypted gradients at the same time. Gradients are mainly aggregated through the addition operation. Paillier encryption \cite{paillier1999public} is one of the state-of-the-art algorithms in PHE. The important operations are listed below.
\begin{itemize}[leftmargin=*]
    \item \textbf{\emph{PaillierSetup$(1^\lambda)\to(PPK, SK)$}:} In the setup, the algorithm randomly chooses two large prime numbers $p_1$ and $p_2$ to compute the public key and private key, which are denoted as $PPK$ and $SK$, respectively.
    \[
    \begin{split}
        PPK = (&n=p_1p_2, h=n+1), \\
        SK = (&\tau=lcm(p_1-1, p_2-1), \\
        &\mu=(L(h^\tau\bmod n^2))^{-1}\bmod n),
    \end{split}
    \]
    where $lcm(p_1-1, p_2-1)$ indicates the least common multiple of $p_1-1$ and $p_2-1$. Here, the function $L(x)$ is defined as $L(x)=\frac{x-1}{n}$.
    
    \item \textbf{\emph{PaillierEncryption$(PPK, m)\to (CT_{m})$}:} The Paillier encryption algorithm takes as input the Paillier public key $PPK$ and the plaintext $m$. Note that $m$ may be gradients, model parameters, or noise in our proposed AerisAI. It outputs the ciphertext $CT_m$ after encrypted with $PPK$.
    \[
        CT_m = \rho_m^{n}h^{m}\bmod n^2,\\
    \]
    where $\rho\in\mathbb{Z}_n$ is a random number for each encryption.
    
    \item \textbf{\emph{Add$(CT_{m_1}, CT_{m_2})\to (CT_{m_{12}})$}:} The function takes as input two pieces of ciphertext, which are encrypted with Paillier encryption. It multiplies these two pieces of ciphertext to support the addition of the two pieces of plaintext, $m_1$ and $m_2$. This function finally outputs the result $CT_{m_{12}}$ as follows.
    \[
    \begin{split}
        CT_{m_{12}} &= CT_{m_1} * CT_{m_2} = \rho_{m_1}^{n}h^{m_1} \rho_{m_2}^{n}h^{m_2} \bmod n^2 \\
        &= (\rho_{m_1}\rho_{m_2})^n(h^{m_1+m_2}) \\
        &= PaillierEncryption(PPK, m_1+m_2)
    \end{split}
    \]
    
    \item \textbf{\emph{Multiply$(CT_m, k)\to (CT_{mk})$}:} This function takes as input a ciphertext and a constant $k$. Next, this function raises the ciphertext $CT_m$ to a constant $k$ power to compute $k$ times of the plaintext $m$. The process is as follows.
    \[
    \begin{split}
        CT_{mk} &= CT_m^k = (\rho_m^{n}h^{m})^k \bmod n^2 = \rho_m^{nk}h^{mk} \bmod n^2 \\
        &= PaillierEncryption(PKK, mk)
    \end{split}
    \]
    \item \textbf{\emph{PaillierDecryption$(SK, CT_m)\to (m)$}:} The algorithm takes as input the Paillier private key $SK$ and the ciphertext $CT_m$. This algorithm decrypts $CT_m$ with $SK$, and it then outputs the original plaintext $m$. The procedure of decryption is described as follows.
    \[
    \begin{split}
        m &= (L(CT_m^\tau \bmod n^2))(\mu\bmod n^2)\bmod n \\
        &= \frac{L(\rho_m^{n}h^{m}\bmod n^2)}{L(h^\tau\bmod n^2)}\bmod n
    \end{split}
    \]
\end{itemize}

\begin{figure*}[t]
    \centering
    \includegraphics[scale=0.55]{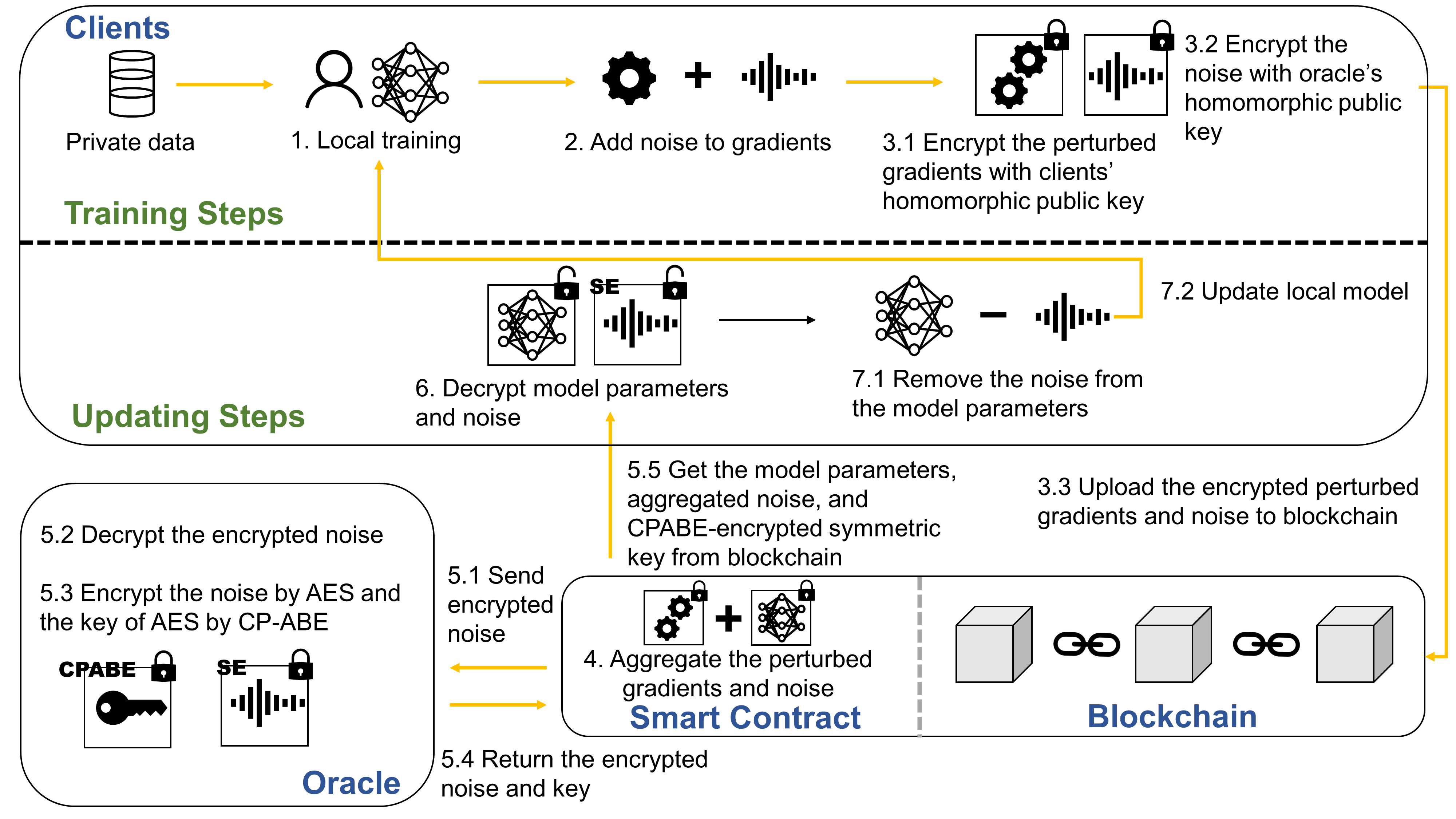}
    \caption{Overview of the proposed \emph{AriesAI}.} \label{fig:workflow}
\end{figure*}

\subsection{System Overview}
Fig. \ref{fig:workflow} shows how the proposed system work. We introduce the blockchain network and smart contract for addressing the weaknesses W1 and W2. Each client employs its own private data to train its local model and then obtains the gradients from the training process. Due to the weakness W3, the clients cannot directly upload the gradients to update the global model. The clients should add the noise to the gradients to achieve differential privacy. After injecting the noise, the clients employ the shared public key of homomorphic encryption to encrypt their perturbed gradients, which enhances the protection of gradients. The details are shown in \emph{S2) Add noise to gradients} and \emph{S3) Encrypt and upload} in Sec. \ref{subc: workflow}. Therefore, each client only sees the perturbed gradients of each other, but not the real gradients. Our system brings noise to overcome the weakness in \cite{guo2022sandbox}, where the clients are able to know the values of gradients with each other.

Although noise is injected into the gradients for privacy preservation, the noise significantly degrades the performance of the global model. As this is one of the major consequences of addressing weakness W3, we name this weakness of performance degradation after noise injection as \emph{\textbf{W3.a} Degradation of the global model performance.} To address \emph{\textbf{W3.a}}, we propose to let the clients encrypt the noise and upload them via the smart contract. Since the noise is encrypted with the oracle's public key of homomorphic encryption, the smart contract can aggregate the encrypted noise, which allows the clients to download the encrypted aggregated noise and remove them from the global model. 

The oracle helps the clients when they download the encrypted noise. The oracle decrypts the encrypted aggregated noise first and then encrypts them with a random session key ($SEK$) of the symmetric encryption. This session key is also necessarily encrypted before sending to the clients. If the oracle encrypts the aggregated noise and session key individually for each client, it is not an efficient and scalable approach. Again, as this weakness is a consequence of addressing weakness W3, we name this scalability issue of noise and session key encryption as \emph{\textbf{W3.b} scalability of encryption}. To address \emph{\textbf{W3.b}}, we introduce group key management based on ciphertext policy attribute-based encryption (CP-ABE) so that the oracle can efficiently distribute the session key ($SEK$) used for decrypting the aggregated noise to the clients. This significantly reduces the transmission cost of distributing noise in our proposed AerisAI. The details of addressing weaknesses W3.a and W3.b are shown in \emph{S3) Encrypt and upload} and \emph{S5) Download}, respectively in Sec. \ref{subc: workflow}. After the clients obtain the global model, aggregated noise, and session keys encrypted by CP-ABE, they are able to decrypt them with corresponding keys. Finally, the clients acquire the global model, and FL completes.

\subsection{Workflow}
\label{subc: workflow}

\begin{itemize}
    \item [\emph{S0)}] \emph{Prestage:} In the system initialization, we assume that the trusted authority calls \textbf{\emph{PaillierSetup$(1^\lambda)$}} to generate two key pairs of Paillier encryption for each client $P_i$ ($i\in\mathbb{Z}_N$, $N$ is the number of clients) and the oracle. The shared key pair for all the clients is denoted as $(PPK_c, SK_c)$. The oracle's public and private key pair is denoted as $(PPK_o, SK_o)$. On the other hand, the trusted authority calls \textbf{\emph{Setup$(1^\lambda)$}} to generate the public key $PK$ and master key $MK$ of CP-ABE. $PK$ is assigned to the oracle. $MK$, which comes from the trusted authority, is used to generate each client's private key $PrivKey_{P_i}$ via calling \textbf{\emph{KeyGeneration$(MK, \mathcal{A}_{P_i})$}}, where $\mathcal{A}_{P_i}$ is the attributes of $P_i$. The smart contract randomly initializes the global model parameters $\theta_G^{(0)}$. Next, it calls \textbf{\emph{PaillierEncryption$(PPK_c, \theta_G^{(0)})$}} to encrypt $\theta_G^{(0)}$ with $PPK_c$, where the encrypted global model is denoted as $CT_{\theta_G^{(0)}}$ and is stored in the smart contract. The noise $\zeta_G^{(0)}$ is initially set as $0$ and encrypted with $PPK_o$ to generate $CT_{\zeta_G^{(0)}}$ by calling \textbf{\emph{PaillierEncryption$(PPK_o, \zeta_G^{(0)})$}}. Each client $P_i$ in the blockchain has its own private dataset $D_i$ for local training. 
    
    \item [\emph{S1)}] \emph{Local training:} At each round $t$, every client trains its own local model whose parameters are initialized with $\theta_G\setminus^{(t)}$. $P_i$ feeds $D_i$ into $\theta_G\setminus^{(t)}$ to compute the loss of the data samples. The clients run the specific gradient descent algorithm for several epochs to update $\theta_G^{(t)}$. Next, $P_i$ gets the new model parameters $\theta_i^{(t+1)}$ to compute the gradients $\delta_i^{(t+1)}$,  which is the difference between $\theta_i^{(t+1)}$ and $\theta_G^{(t)}$ (i.e., $\delta_i^{(t+1)}=\theta_i^{(t+1)}-\theta_G^{(t)}$).
    
    \item [\emph{S2)}] \emph{Add noise to gradients:} After getting $\delta_i^{(t+1)}$, each client randomly chooses the noise $\zeta_i^{(t+1)}$ from the Gaussian distribution \cite{xu2022spdl} with the mean $\mu$  and variance $\sigma$. $P_i$ obtains the new perturbed gradients $\hat{\delta}_{i}^{(t+1)}$ after adding $\zeta_i^{(t+1)}$ to $\delta_i^{(t+1)}$ (i.e., $\hat{\delta}_{i}^{(t+1)}=\delta_i^{(t+1)}+\zeta_i^{(t+1)}$).
    
    \item [\emph{S3)}] \emph{Encrypt and upload:} $P_i$ encrypts perturbed gradients $\hat{\delta}_{i}^{(t+1)}$ and noise $\zeta_i^{(t+1)}$ with different public keys. For encrypting $\hat{\delta}_{i}^{(t+1)}$ with the client's public key of Paillier encryption, each client calls \textbf{\emph{PaillierEncryption$(PPK_c, \hat{\delta}_{i}^{(t+1)})$}} to obtain the encrypted result, which is denoted as $CT_{\hat{\delta}_{i}^{(t+1)}}$. Besides, $\zeta_i^{(t+1)}$ is encrypted with the oracle's public key by calling \textbf{\emph{PaillierEncryption$(PPK_o, \zeta_i^{(t+1)})$}} to obtain the encrypted noise $CT_{\zeta_i^{(t+1)}}$. Next, $P_i$ uploads $CT_{\hat{\delta}_{i}^{(t+1)}}$ and $CT_{\zeta_i^{(t+1)}}$ via the smart contract.
    
    \item [\emph{S4)}] \emph{Aggregation:} The smart contract aggregates the encrypted perturbed gradients and noise that are uploaded by the clients. First, the smart contract iteratively calls \textbf{\emph{Add$(CT_{\hat{\delta}_{i}^{(t+1)}}, CT_{\hat{\delta}_{i+1}^{(t+1)}})$}} ($\forall{i}\in\mathbb{Z}_N$) to compute the sum of all the clients' encrypted perturbed gradients $CT_{\delta_S^{(t+1)}}$. Then, it calls \textbf{\emph{Multiply$(CT_{\hat{\delta}_S^{(t+1)}}, \frac{1}{N})$}} to get the mean $CT_{\hat{\delta}_{M}^{(t+1)}}$ of $CT_{\hat{\delta}_{S}^{(t+1)}}$.
    Afterward, the smart contract invokes \textbf{\emph{Add$(CT_{\theta_G^{(t)}}, CT_{\delta_{M}^{(t+1)}})$}} to update the encrypted global model $CT_{\theta_G^{(t)}}$. Aggregating encrypted noise is similar to the process of aggregating encrypted perturbed gradients. The updated noise $CT_{\zeta_G^{(t+1)}}$ is then obtained after aggregation, which can be used to remove the total noise in each client for achieving better model accuracy.
    
    \item [\emph{S5)}] \emph{Download:} At the next round $t+1$, the clients directly download the updated encrypted global model $CT_{\theta_G^{(t+1)}}$ via the smart contract for iterative training or specific applications. To download the encrypted noise, the oracle helps assign them to the clients launching the request. The oracle invokes \textbf{\emph{PaillierDecryption$(SK_o, CT_{\zeta_G^{(t+1)}})$}} to obtain the noise $\zeta_G^{(t+1)}$. Then, $\zeta_G^{(t+1)}$ is encrypted by the symmetric-key algorithm (AES in our proposed AerisAI). Besides, the session key $SEK$ of AES is necessarily encrypted with the specific policy $\mathcal{P}$ under CP-ABE cryptosystem by calling \textbf{\emph{Encryption$(PK, SEK, \mathcal{P})$}} to acquire $CT_{SEK}$. \emph{AES$(\zeta_G^{(t+1)})$} and $CT_{SEK}$ are broadcasted to the clients. Thanks to CP-ABE, which is a broadcast encryption, our proposed AerisAI is suitable to broadcast the required parameters to multiple clients. 
    
    \item [\emph{S6)}] \emph{Decrypt model and noise:} In this step, the clients call \textbf{\emph{PaillierDecryption$(SK_c, CT_{\theta_G^{(t+1)}})$}} to get the decrypted global model parameters $\theta_G^{(t+1)}$. Note that $\theta_G^{(t+1)}$ contains the noise. After that, the clients need to decrypt $CT_{SEK}$ and \emph{AES$(\zeta_G^{(t+1)})$} for removing the noise from the decrypted global model parameters. \textbf{\emph{Decryption$(CT_{SEK}, PrivKey_{P_i})$}} is called to retrieve the session key $SEK$ by each client $P_i$. $SEK$ is used to decrypt \emph{AES$(\zeta_G^{(t+1)})$} by AES, and then the clients obtain the aggregated noise $\zeta_G^{(t+1)}$.
    
    \item [\emph{S7)}] \emph{Remove noise and update local model:} We exclude the noise from the decrypted global model parameters, i.e., $\theta_G^{(t+1)}-\zeta_G^{(t+1)}$, to allow the clients to obtain the new global model parameters without the noise, i.e., $\theta_{G\setminus}^{(t+1)}$. Finally, each client replaces its local model with $\theta_{G\setminus}^{(t+1)}$ and trains the local model with its private dataset at the next round.
\end{itemize}

\section{Security Analysis}
\label{c:Securityanalysis}
We formally analyze the security of the proposed AerisAI as follows. 

\begin{theorem}
\label{theo: Paillier}
Based on the difficulty of factoring large integers \cite{lenstra2000integer}, our proposed AerisAI maintains the privacy that the encrypted gradients and noise cannot be decrypted without the corresponding private key.
\end{theorem}
\begin{proof}
In our proposed AerisAI, the gradients and noise are encrypted by Paillier encryption. The public key is $PPK=(n, h=n+1)$ where $n$ is multiplied by two large prime numbers $p_1$ and $p_2$. In the following, we show how to reduce our encryption to be as secure as other cryptosystems based on integer factorization. The gradient $\delta$ was encrypted to generate $CT_{\delta}$ by calling \emph{PaillierEncryption$(PPK, \delta)$}. To decrypt $CT_{\delta}$, the adversary needs to recover the private key $SK=(\tau=lcm(p_1-1,p_2-1), \mu=(L(h^\tau \bmod n^2)^{-1}\bmod n)$, which is generated with the prime numbers $p_1$ and $p_2$. The adversary needs to factor $n$ into $p_1$ and $p_2$ and compute $lcm(p_1, p_2)$. However, the existence of a polynomial-time algorithm for factoring large integers on a classical computer is deemed unlikely \cite{lenstra2000integer}. The security is based on the hardness of factoring the large integer $n$, which can be regarded as the integer factoring problem.
\end{proof}

\begin{theorem}
\label{theo:DCR}
The Paillier encryption is semantically secure under the decisional composite residuosity (DCR) assumption \cite{paillier1999public}, so the gradients and noise are secure in our proposed AerisAI when the DCR assumption holds.
\end{theorem}
\begin{proof}
The gradient $\delta_1$ is encrypted as $CT_{\delta_1}=\rho_{\delta_1}^{n}h^{\delta_1}\bmod n^2$. We define a number $z$ is said to be the $n-$th residue modulo $n^2$ if there exists a number $\rho_{\delta_1}\in\mathbb{Z}_{n^2}$ such that $z = \rho_{\delta_1}^{n}\bmod n^2$. Based on the DCR assumption \cite{paillier1999public}, it is intractable for deciding whether there exists a $\rho_{\delta_{1}}$ in polynomial time such that $z = \rho_{\delta_1}^{n}\bmod n^2$. Therefore, an adversary cannot distinguish whether two encrypted gradients, $CT_{\delta_1}$ and $CT_{\delta_2}$, represent the same gradients (i.e., $\delta_1=\delta_2$) and thus they are not able to learn information about the gradients. The security of Paillier encryption is based on the computational infeasibility of the DCR assumption, which ensures the privacy preservation.


\end{proof}

\begin{theorem}
\label{theo: CPABE}
Based on the security of CP-ABE, the encrypted session key is collusion-resistant against the colluding clients.
\end{theorem}
\begin{proof} 
In CP-ABE, the message $m$ is encrypted as $\Tilde{C}=m\cdot e(g,g)^{\alpha s})$, where $e:G_1\times G_1\to G_T$, $\alpha,s\in\mathbb{Z}^*_p$. We also encrypt the session key of AES $SEK$ with the public key $PK$ and policy $\mathcal{P}$ based on CP-ABE. By calling $Encryption(PK, SEK, \mathcal{P})$, we obtain the encrypted session key $\Tilde{C}_{SEK}=SEK\cdot e(g,g)^{\alpha s}$ (i.e., the message $m$ in the original CP-ABE scheme is replaced with the session key $SEK$). Therefore, the attacker who wants to get $SEK$ must match $l_{f}(0)$ in $C_f$ from $CT_{SEK}$ with $S$ components from colluding the user's private key for an attribute $i$ that the attacker does not have, to recover $e(g,g)^{\alpha s}$. However, $e(g,g)^{\alpha s}$ is blinded by the random value $e(g,g)^{rs}$, where $r$ is a unique random number for each client. Therefore, the attacker cannot obtain the $SEK$ by a collusion attack because of the blinding random value from each user's private key.
\end{proof}
\begin{theorem}
The oracle on the blockchain cannot break privacy to recover the information of the clients' private data with its key pair.
\end{theorem}
\begin{proof}
In our proposed AerisAI, we assume the oracle is honest-but-curious, which means that the oracle may want to recover the client's private data from the gradients. The oracle's public and private key pair of Paillier encryption $(PPK_o, SK_o)$ is independent of the client's key pair of Paillier encryption $(PPK_c, SK_c)$. Therefore, the oracle cannot derive $SK_c$ from $SK_o$. Based on Theorem \ref{theo: Paillier}, the oracle cannot decrypt the encrypted perturbed gradients $\hat{\delta}_{i}^{(t+1)}$ uploaded by the clients with an efficient (polynomial-time) algorithm. As a result, the oracle can only decrypt the encrypted global noise $CT_{\zeta_G^{(t+1)}}$ and individual noise $CT_{\zeta_i^{(t+1)}}$ each client randomly chooses from the Gaussian distribution. However, those data do not include the information about the clients' data.
\end{proof}

\section{Evaluation}
\label{c:Evaluatoin}
We evaluate the effectiveness of the proposed AerisAI with other baseline approaches in terms of functionality, model accuracy, and efficiency.

\subsection{Functionality Comparisons}
For functionality comparisons, we compare our proposed \texttt{AerisAI} with the state-of-the-art FL schemes, including \texttt{PrivacyDL+HE}~\cite{aono2017privacy}, \texttt{SGX}~\cite{kalapaaking2022blockchain}, \texttt{Sandbox}~\cite{guo2022sandbox}, \texttt{SPDL}~\cite{xu2022spdl}, and \texttt{DS2PM}~\cite{chen2021ds2pm}. \texttt{PrivacyDL+HE} \cite{aono2017privacy} is a FL framework with additively homomorphic encryption to protect the gradients. The other compared baselines~\cite{kalapaaking2022blockchain, guo2022sandbox, xu2022spdl, chen2021ds2pm} are the blockchain-based collaborative AI framework and introduce different methods to protect the data privacy. Table \ref{table: comparison} compares our proposed \texttt{AerisAI} with other baselines, where our proposed AerisAI fulfills \emph{all the security requirements}. Please note that the security requirements in Table \ref{table: comparison} correspond to the weaknesses W1, W2, W3, and W3.a identified above in this paper. The comparisons of how our proposed AerisAI and the other baselines address weakness W3.b, i.e., \emph{scalability} will be evaluated in Sec. \ref{c:efficiency}.

\subsection{Performance Evaluation}\label{c:perf_comp}
In addition to the functionality comparisons above, we compare the performance of the proposed AerisAI in terms of model accuracy and efficiency with the state-of-the-art baselines. 

\begin{table*}[ht]
    \scriptsize
    \caption{\label{table: comparison} Functionality comparisons.}
    \centering
    \begin{threeparttable}
    \resizebox{\linewidth}{!} 
    {
        \setlength{\tabcolsep}{7mm} 
        {
        \begin{tabular}{cccccc}
            \toprule[2pt]
            \diagbox{Scheme}{Function} & \textbf{Decentralization} & \textbf{Auditability} & \begin{tabular}[c]{@{}l@{}}\textbf{Gradients} \\ \textbf{perturbation}\end{tabular} & \begin{tabular}[c]{@{}l@{}}\textbf{Gradients}\\ \textbf{masking}\end{tabular} & \begin{tabular}[c]{@{}l@{}}\textbf{Group key} \\ \textbf{management}\end{tabular}\\ \hline
            \texttt{PrivacyDL+HE}~\cite{aono2017privacy} & - & - & - & \checkmark & -     \\ 
            \texttt{SGX}~\cite{kalapaaking2022blockchain} & \checkmark & \checkmark & - & - & -     \\ 
            \texttt{Sandbox}~\cite{guo2022sandbox} & \checkmark & \checkmark & - & - & -     \\ 
            \texttt{SPDL}~\cite{xu2022spdl} & \checkmark & \checkmark & \checkmark & - & -     \\ 
            \texttt{DS2PM}~\cite{chen2021ds2pm} & \checkmark & \checkmark & - & \checkmark & -     \\ 
            
            \textbf{\texttt{AerisAI (Ours)}} & \textbf{\checkmark} & \textbf{\checkmark} & \textbf{\checkmark} & \textbf{\checkmark} & \textbf{\checkmark}     \\ 
            \bottomrule[2pt]
        \end{tabular}
        }
    }
    \end{threeparttable}
\end{table*}

\subsubsection{Experimental Settings}
\phead{Blockchain.} We create a blockchain network with Hyperledger Fabric \cite{androulaki2018hyperledger}, a widely adopted architecture adopted by many research works~\cite{sousa2018byzantine, sharma2018databasify, sharma2019blurring}, to evaluate our proposed AerisAI. Hyperledger Fabric focuses on consortium networks that enable enterprises and organizations to support their needs.

\phead{Datasets.} We evaluate the model performance on two widely-adopted benchmark datasets, i.e., \emph{MNIST}~\cite{deng2012mnist}, \emph{CIFAR-10}~\cite{krizhevsky2009learning}, and \emph{CIFAR-100} \cite{krizhevsky2009learning}. MNIST is a handwritten digits dataset that consists of 60,000 training data and 10,000 testing data with 10 classes. CIFAR-10 is an object recognition dataset that consists of 50,000 training data and 10,000 testing data with 10 classes, including airplanes, automobiles, birds, cats, deers, dogs, frogs, horses, ships, and trucks. CIFAR-100 is similar to CIFAR-10, except it has 100 classes. CIFAR-100 also consists of 50,000 training data and 10,000 testing data, so each class contains 600 images.   

\phead{Data partition, model architecture, and training.} 
For the two datasets above, we partition and distribute the training data to all the clients equally and randomly. We evaluate the accuracy of the global model with the testing data. The local and global models are Multi-layer Perceptrons (MLPs) for MNIST. Besides, we use SqueezeNet \cite{iandola2016squeezenet} as the models for CIFAR-10 and CIFAR-100. For MNIST, the local and global models are MLPs with dimensions 784 (input layer), 128 (hidden layer), 64 (hidden layer), and 10 (output layer). For CIFAR-10 and CIFAR-100, each client trains the lightweight SqueezeNet with 6x fewer parameters to classify the image categories due to the storage limit in Hyperledger Fabric. The maximum block size is 100 MB \cite{HyperDoc}, and any transaction larger than this value will be rejected. In our experiment, each transaction, which includes the encrypted model parameters and noises, is approximately 98 MB. Despite using the lightweight SqueezeNet to evaluate performance, we have integrated our proposed AerisAI with Hyperledger Fabric, making it potentially applicable in practical scenarios. All models are trained by each client with Adam optimization~\cite{kingma2014adam}. The learning rate is set as $0.001$ for MNIST and CIFAR10, and $0.01$ for CIFAR100.. The experiments are conducted on a server with an AMD Ryzen 9 5900X@3.7GHz CPU, NVIDIA 3080 Ti GPU, and 48GB main memory. 

\phead{Baselines.}
To demonstrate the effectiveness of our proposed AerisAI, we compare our work with four well-known or state-of-the-art baselines. i) \texttt{Local Training}. Each client trains the model with only its private dataset. In other words, the clients do not collaboratively train a global model. ii) \texttt{Standard Averaging Federated Learning (SAFL)}~\cite{mcmahan2017communication}. This is the traditional FL framework that the server updates the global model by computing the average of the clients' gradients. iii) \texttt{Centralized Training}. This is the centralized machine learning algorithm, which can be viewed as the upper bound of the model accuracy. All the data are merged at the server, and the centralized server leverages the whole data to train the global model. In other words, the clients directly share their data without considering any data privacy issues. The clients do not collaboratively train a global model in the FL framework. iv) \texttt{SPDL}~\cite{xu2022spdl}. This is the state-of-the-art decentralized collaborative AI framework.


Please note that the accuracy evaluation does not include the baselines \texttt{PrivacyDL+HE}, \texttt{Sandbox}, and \texttt{SGX} that were compared in the functionality comparisons. This is because the overall framework of \texttt{PrivacyDL+HE} is very similar to \texttt{SAFL}, but \texttt{PrivacyDL+HE} employs homomorphic encryption, which achieves similar model accuracy but is significantly slower than \texttt{SAFL}. Therefore, we choose \texttt{SAFL} for comparisons. In addition, \texttt{SGX} and \texttt{Sandbox} do not protect the gradients (i.e., lacking privacy preservation), and thus are not chosen as baselines in our accuracy evaluation. Although \texttt{DS2PM} encrypts the model parameters, the encrypted model parameters are decrypted before aggregation, which results in privacy leakage. Therefore, it is not included in our accuracy evaluations.

\begin{figure*}[ht]
    \centering
    \subfigure[][5 clients.]{
        \label{fig:2a}
        \includegraphics[width=0.6\columnwidth]{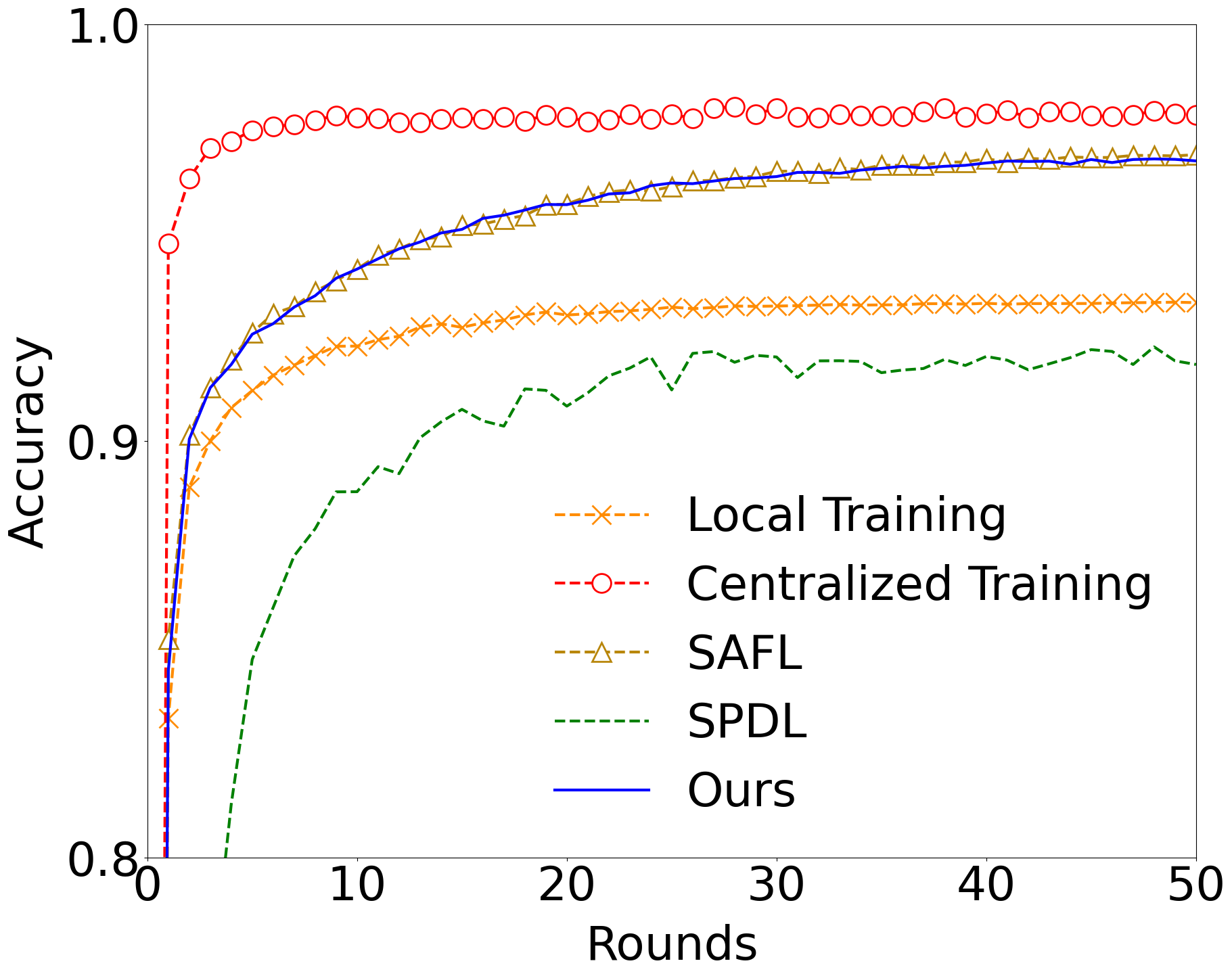}
    }
    \subfigure[][10 clients.]{
        \label{fig:2b}
        \includegraphics[width=0.6\columnwidth]{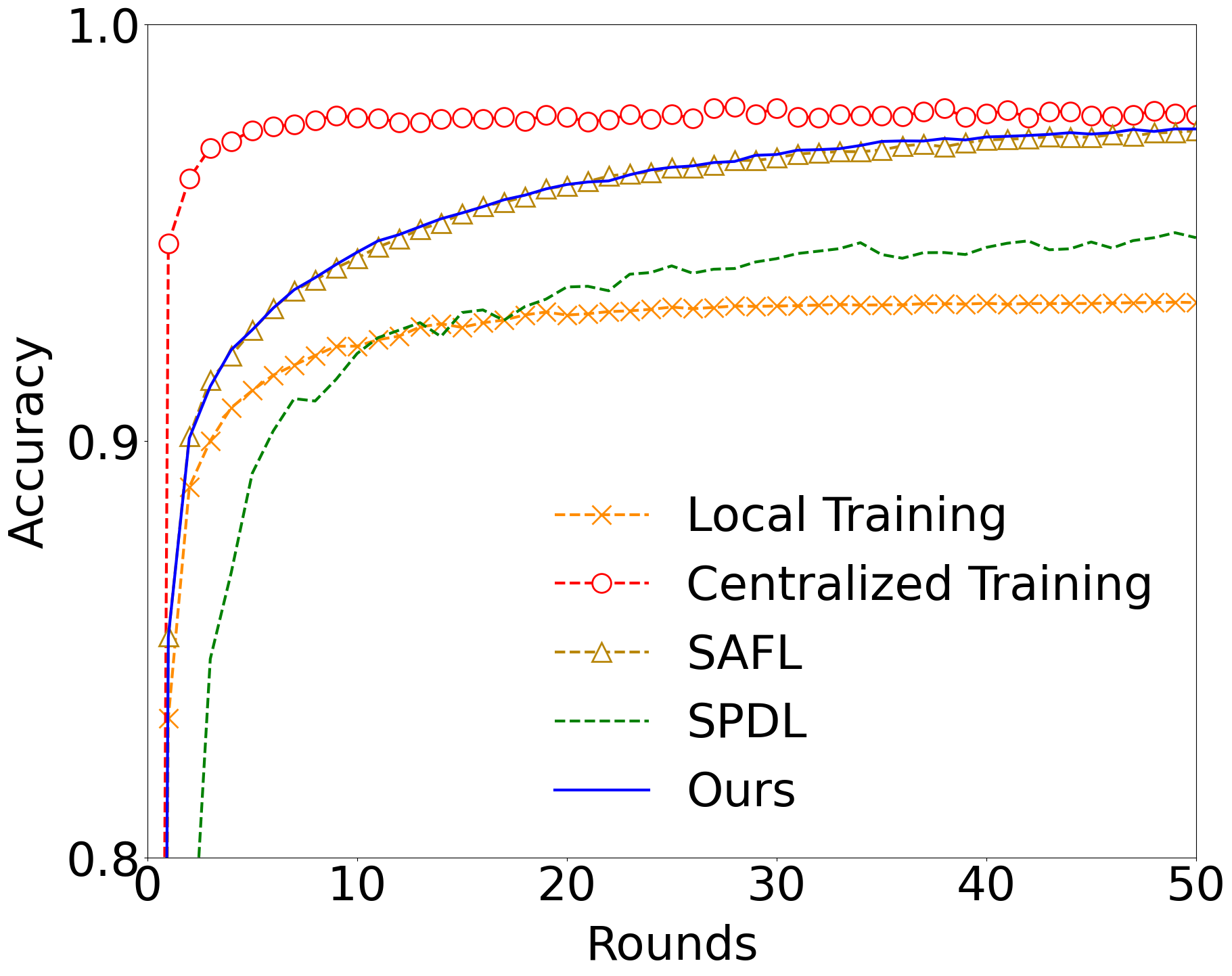}
    }
    \subfigure[][15 clients.]{
        \label{fig:2c}
        \includegraphics[width=0.6\columnwidth]{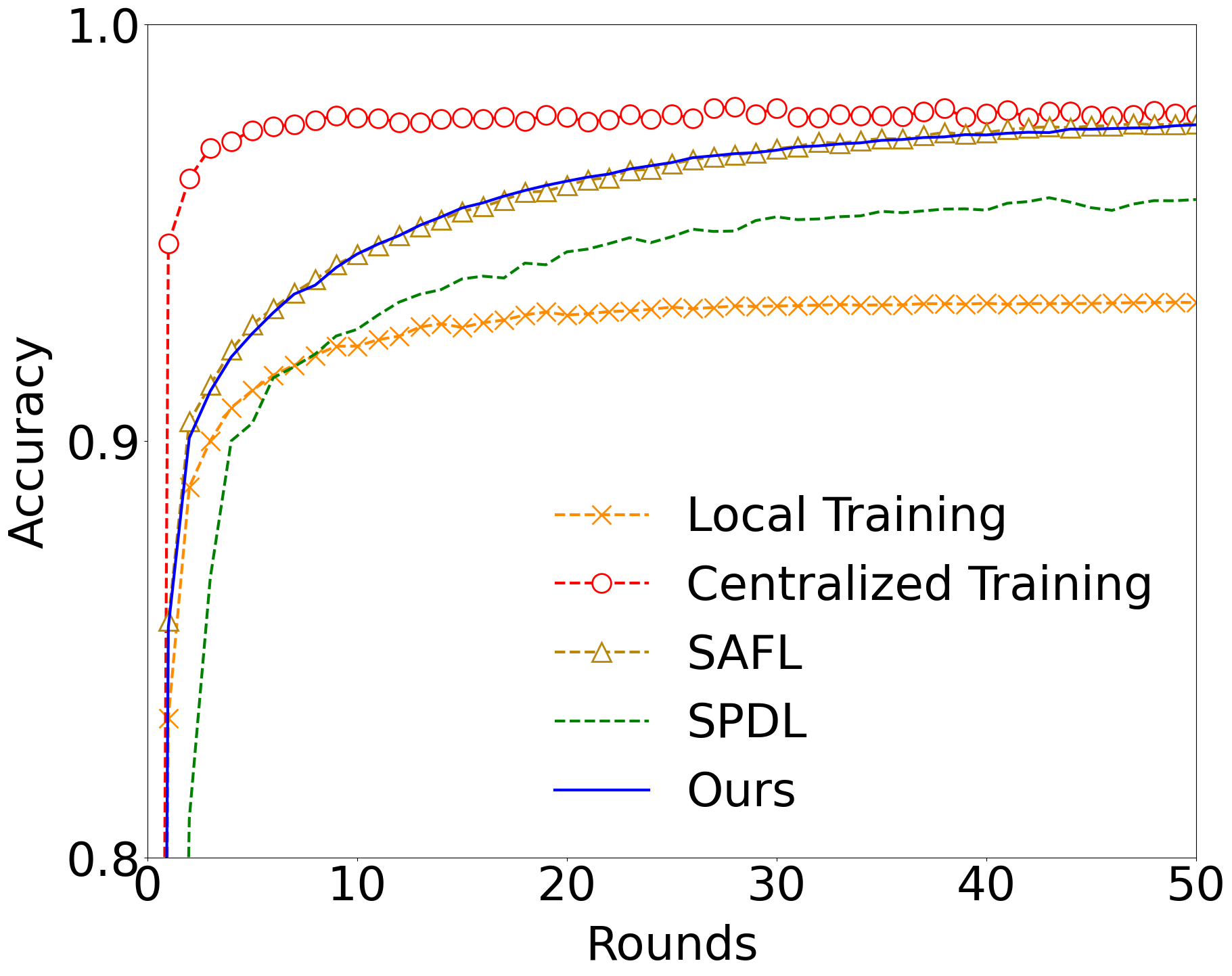}
    }
    \caption{Accuracy of different numbers of clients for \emph{MNIST}.}
    \label{fig:MNIST_RoundAccu}
\end{figure*}

\subsubsection{Model Performance}
We analyze the accuracy of the global models with different schemes. For \texttt{Local Training}, the clients train their local models with only their private data, and thus we evaluate the performance of each local model and compute the average. We present the accuracy on MNIST with different numbers of clients in Fig. \ref{fig:MNIST_RoundAccu}. The x-axis represents the training rounds, where \emph{one round} indicates that the global model is updated with the aggregated gradients once. For \texttt{Local Training}, \emph{one round} is equivalent to one epoch for training the local model. We evaluate the accuracy of the global model with different numbers of clients, i.e., $\{5, 10, 15\}$. \texttt{Local Training} is always inferior to our proposed \texttt{AerisAI} (denoted as \texttt{Ours} in the figure) because each client trains the local model with only its private data without collaboration. The performance of \texttt{AerisAI} (i.e., \texttt{ours}) outperforms \texttt{SPDL} while achieving similar results to \texttt{SAFL}, quite close to \texttt{Centralized Training}. However, \texttt{SAFL} neither achieves decentralization nor protects the gradients, which has the risks of a single point failure and privacy leakage. Please note that \texttt{Centralized Training} does not employ any privacy preservation techniques and does not consider the decentralized factor either. \texttt{Centralized Training} is adopted here as the performance upper bound for the baselines only. In addition, our proposed \texttt{AerisAI} performs well with different numbers of clients, indicating that \texttt{AerisAI} is robust and able to be deployed in different application scenarios. 

\begin{figure*}[ht]
    \centering
    \subfigure[][5 clients.]{
        \label{fig:3a}
        \includegraphics[width=0.6\columnwidth]{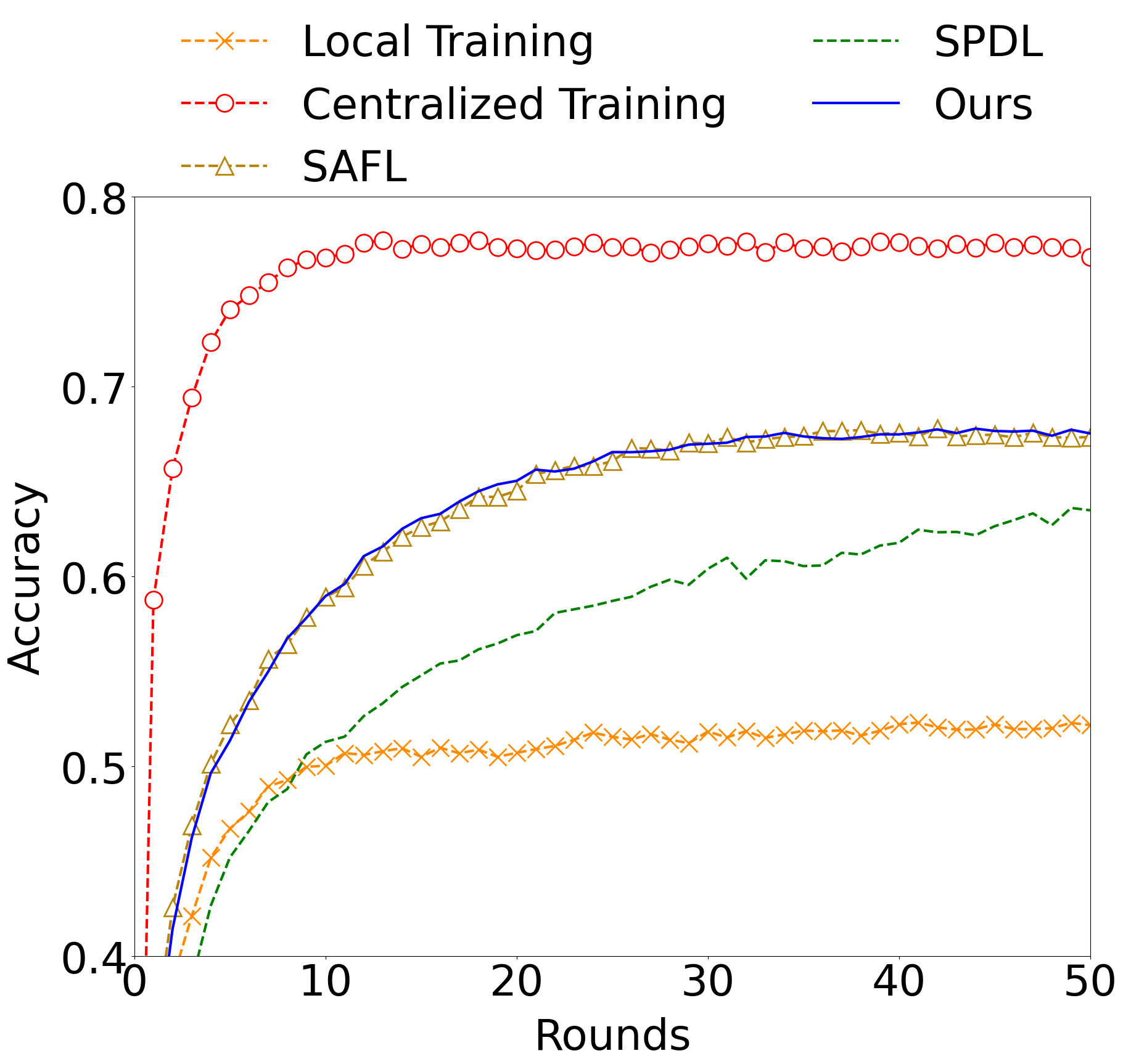}
    }
    \subfigure[][10 clients.]{
        \label{fig:3b}
        \includegraphics[width=0.6\columnwidth]{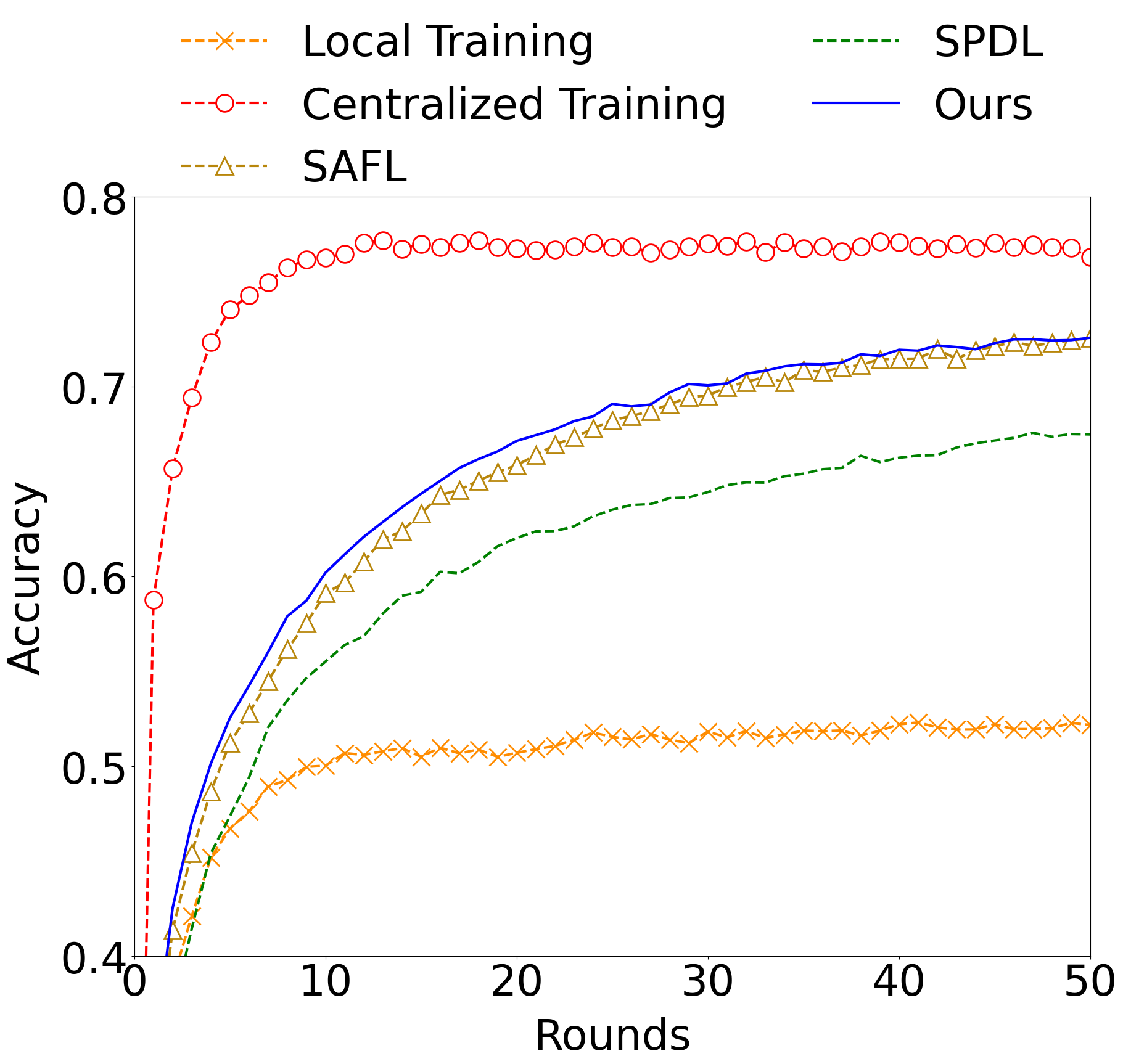}
    }
    \subfigure[][15 clients.]{
        \label{fig:3c}
        \includegraphics[width=0.6\columnwidth]{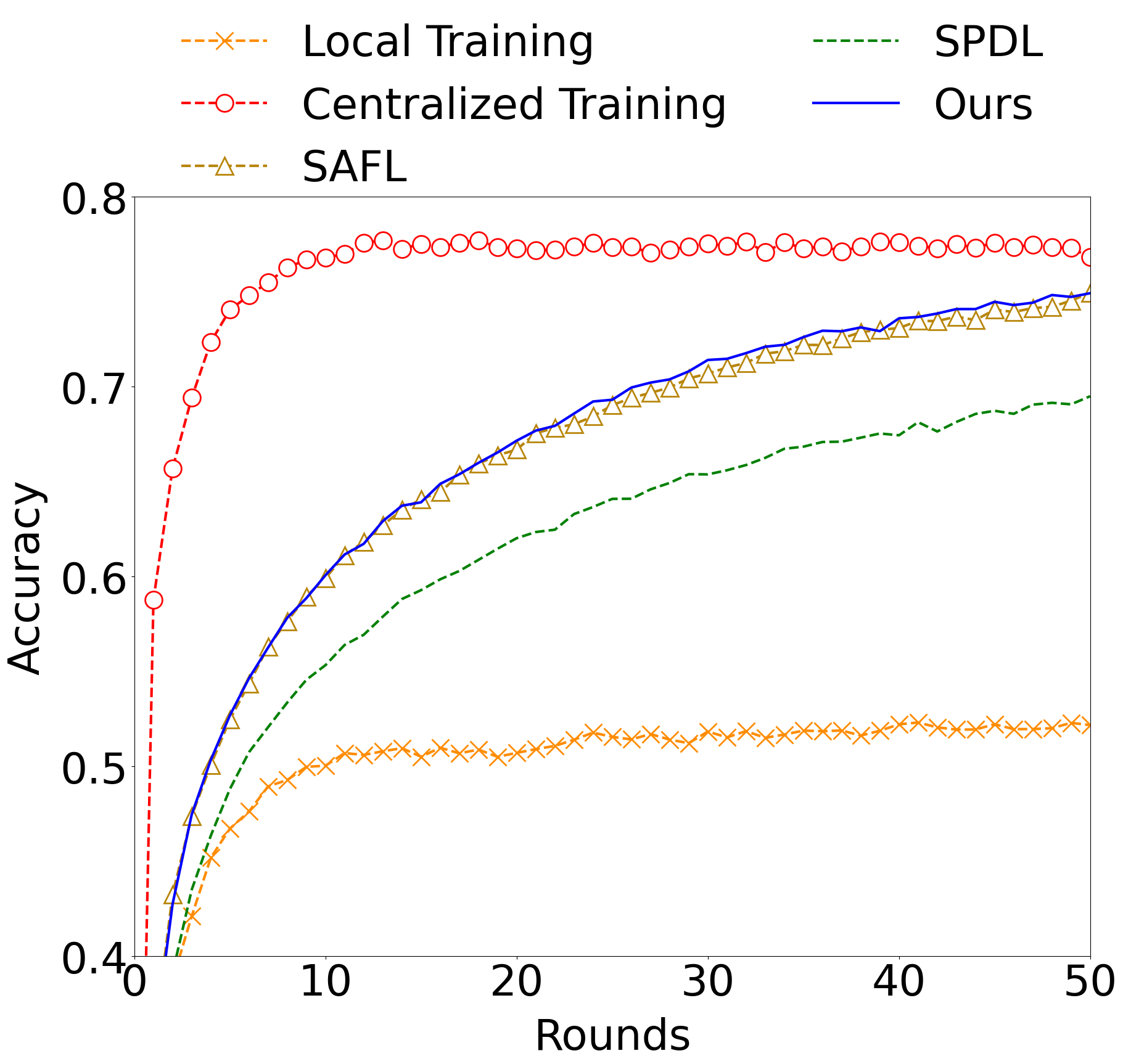}
    }
    \caption{Accuracy of different numbers of clients for \emph{CIFAR-10}.}
    \label{fig:CIFAR10_RoundAccu}
\end{figure*}

\begin{figure*}[ht]
    \centering
    \subfigure[][5 clients.]{
        \label{fig:4a}
        \includegraphics[width=0.6\columnwidth]{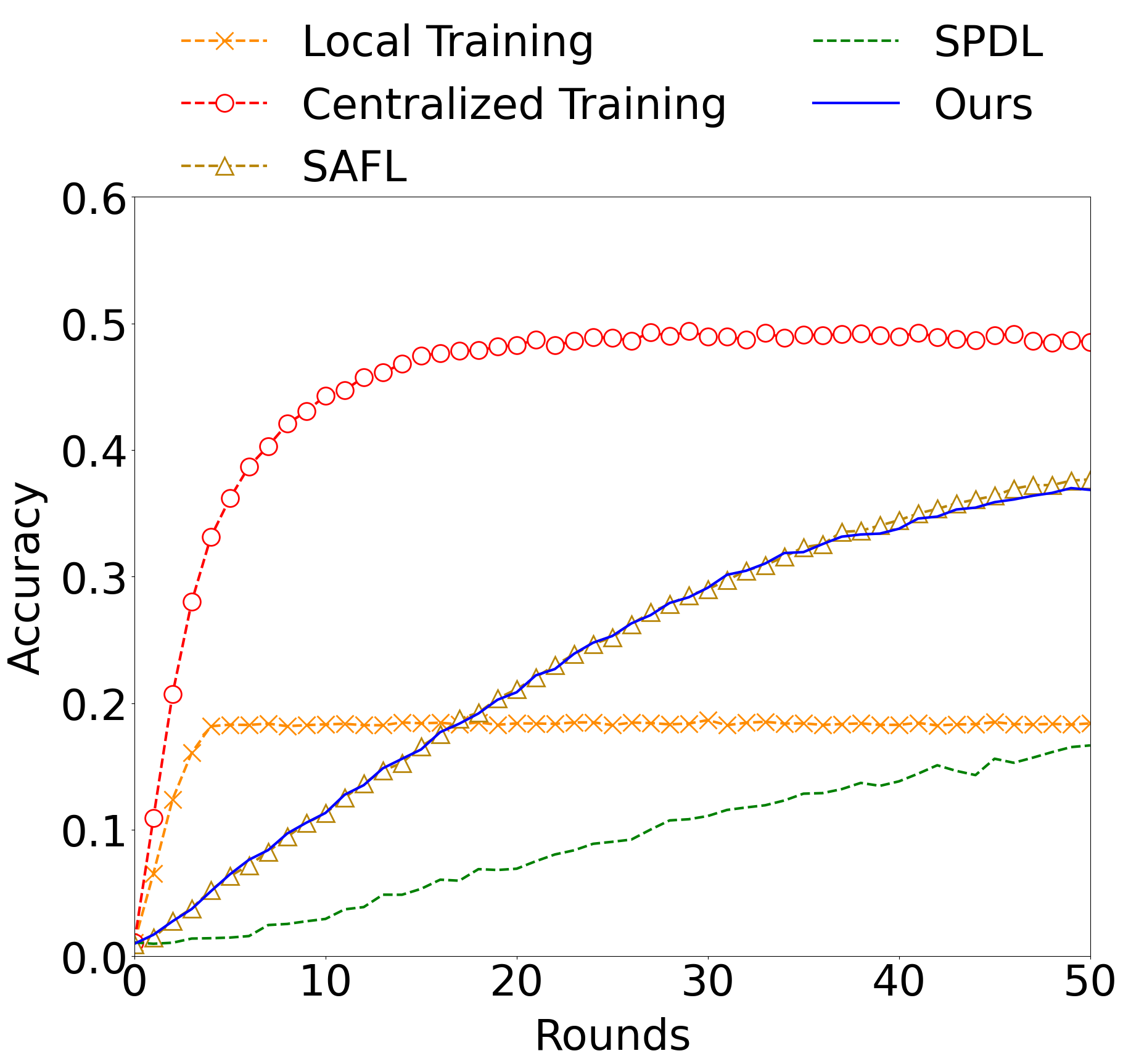}
    }
    \subfigure[][10 clients.]{
        \label{fig:4b}
        \includegraphics[width=0.6\columnwidth]{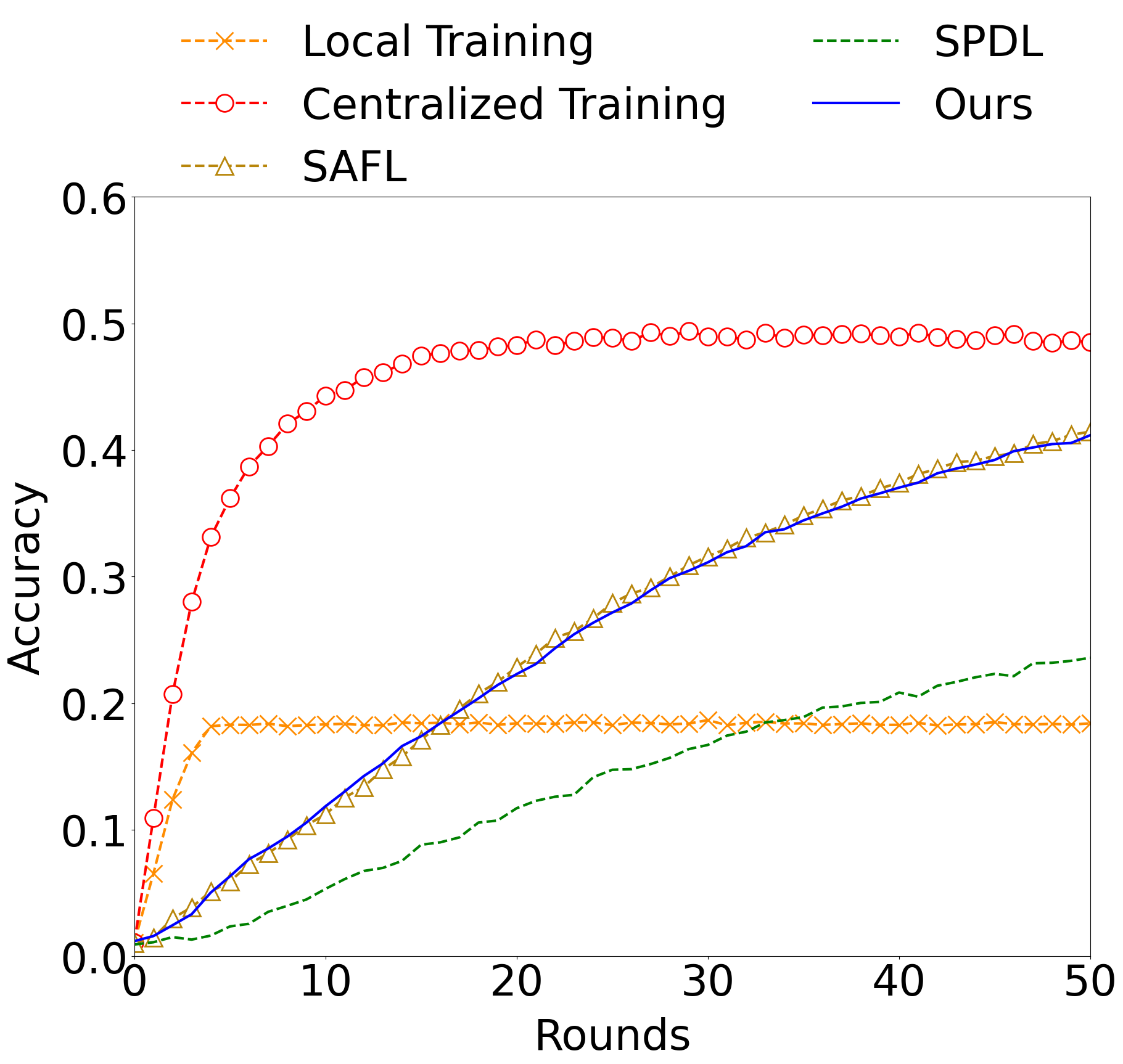}
    }
    \subfigure[][15 clients.]{
        \label{fig:4c}
        \includegraphics[width=0.6\columnwidth]{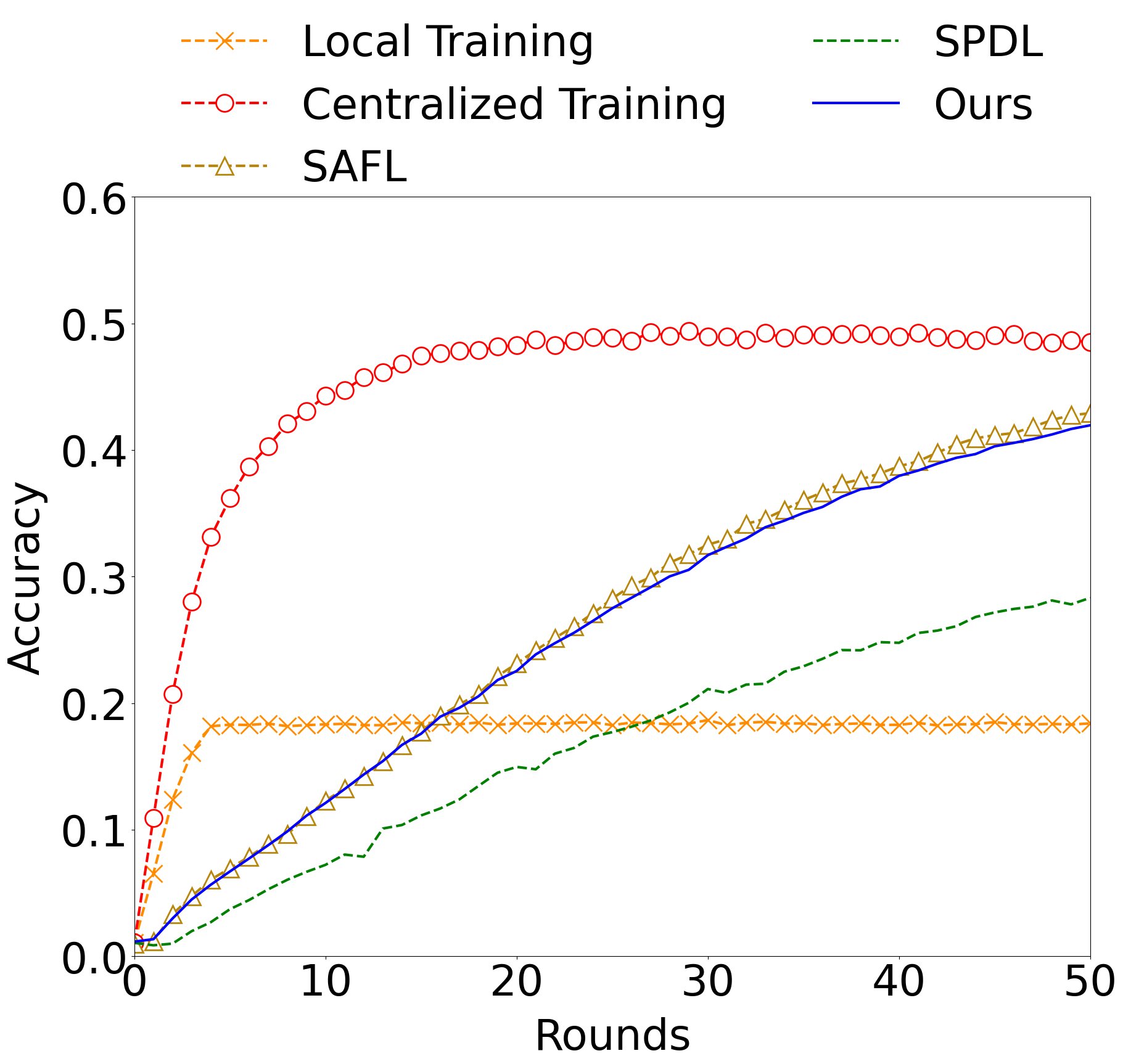}
    }
    \caption{Accuracy of different numbers of clients for \emph{CIFAR-100}.}
    \label{fig:CIFAR100_RoundAccu}
\end{figure*}

For CIFAR-10, the testing results are shown in Fig. \ref{fig:CIFAR10_RoundAccu}. The overall trend of our proposed \texttt{AerisAI} is similar to  the results in MNIST, even if the model is a convolution neural network, i.e., lightweight SqueezeNet. That is, our proposed \texttt{AerisAI} outperforms \texttt{Local Training} and \texttt{SPDL} while achieving comparative performance with \texttt{Centralized Training} with larger rounds. This indicates that \texttt{AerisAI} is general and suitable for a wide spectrum of model architectures. Notice that the Paillier encryption does not degrade the accuracy in our proposed \texttt{AerisAI}, and the performance is almost the same as the scheme without encryption (i.e., \texttt{SAFL}). This also indicates that the smart contract correctly updates the global model. Similarly, \texttt{AerisAI} outperforms \texttt{SPDL} for CIFAR-100, as shown in Fig.\ref{fig:CIFAR100_RoundAccu}. In the scenario where 5 clients collaborate to train the global model, the performance of the global model in \texttt{SPDL} is even worse than that of \texttt{Local Training}. However, the proposed \texttt{AerisAI} performs as well as \texttt{SAFL} while simultaneously ensuring privacy preservation.

\begin{figure*}[ht]
    \centering
    \subfigure[][\emph{MNIST}.]{
        \label{fig:5a}
        \includegraphics[width=0.6\columnwidth]{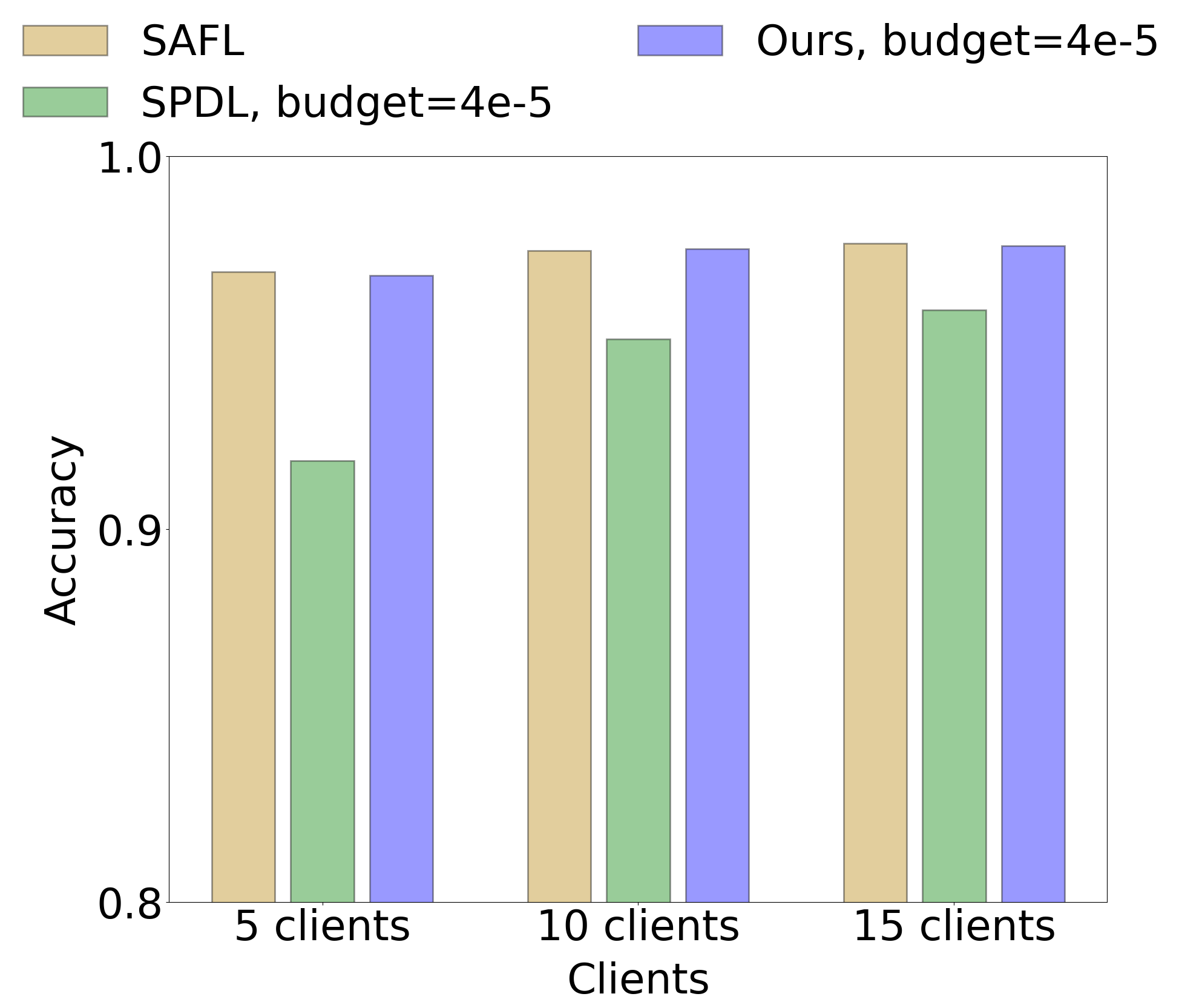}
    }
    \subfigure[][\emph{CIFAR-10}.]{
        \label{fig:5b}
        \includegraphics[width=0.6\columnwidth]{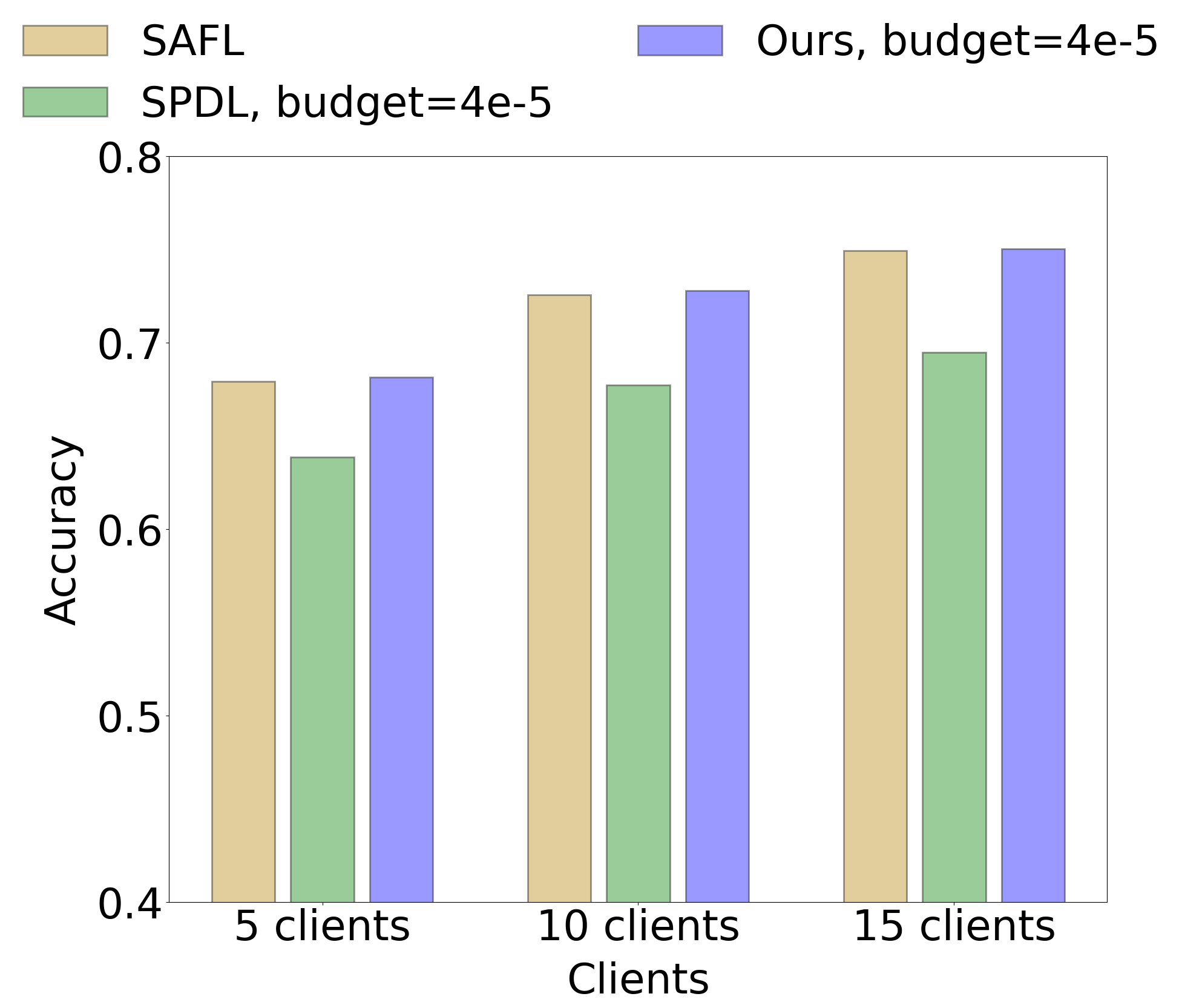}
    }
    \subfigure[][\emph{CIFAR-100}.]{
        \label{fig:5c}
        \includegraphics[width=0.6\columnwidth]{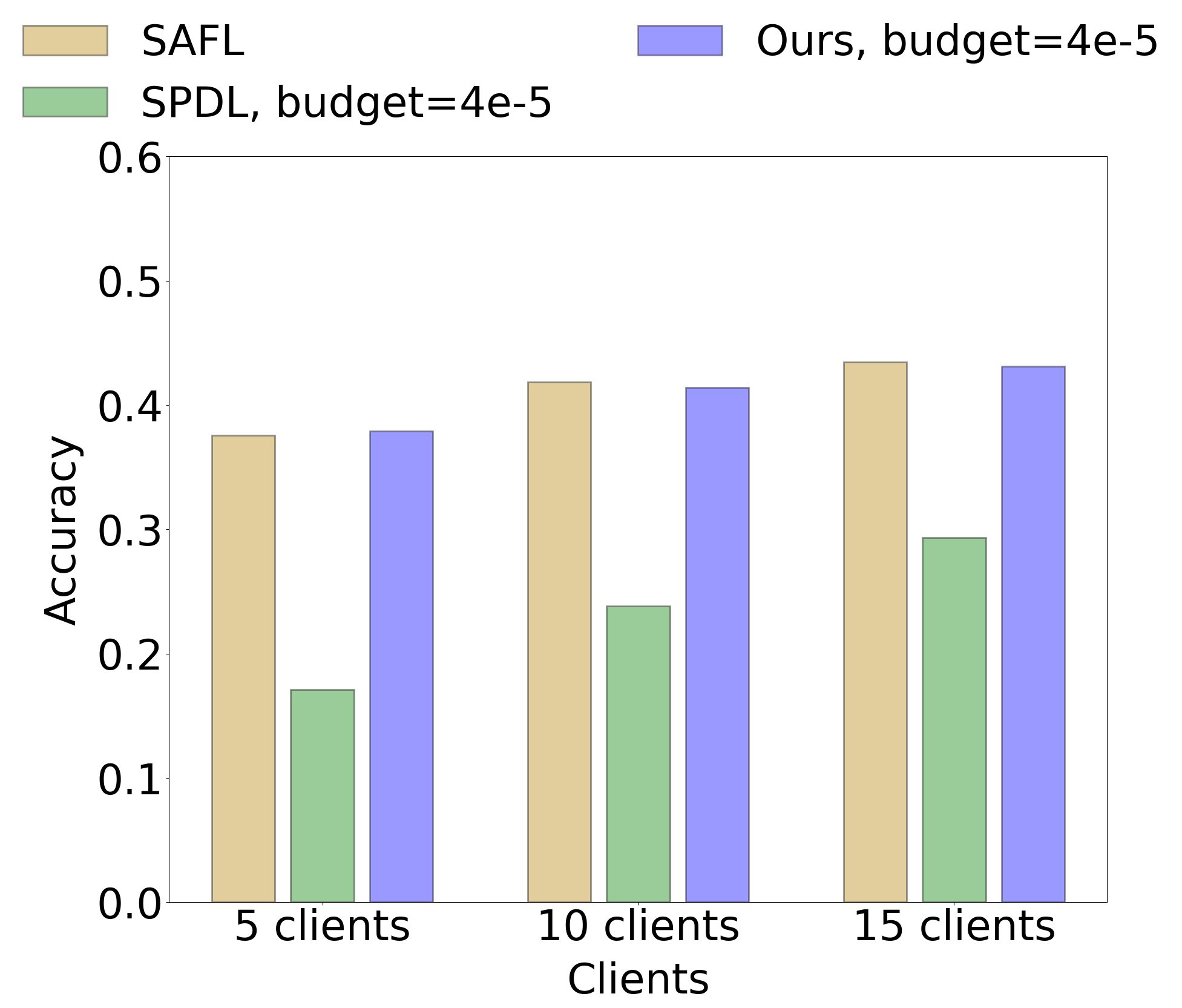}
    }
    \caption{Accuracy of different methods for \emph{MNIST}, \emph{CIFAR-10}, and \emph{CIFAR-100}.}
    \label{fig:threedataset_ClientsAccu}
\end{figure*}

To better understand the baselines that are closely related to our proposed \texttt{AerisAI}, we further present the model performance of \texttt{AerisAI} (i.e., \texttt{Ours}), \texttt{SAFL}, and \texttt{SPDL}. Fig. \ref{fig:threedataset_ClientsAccu} presents the results. The performance of \texttt{Ours} is very close to that of \texttt{SAFL} in both datasets. However, \texttt{SAFL} neither protects the gradients from leakage nor provides differential privacy. 
The accuracy of \texttt{SPDL} is far lower than that of \texttt{AerisAI} (\texttt{Ours}) across all three datasets, particularly for CIFAR-100, which includes a larger number of data classes. In summary, our proposed \texttt{AerisAI} achieves a good balance of model performance and data privacy, outperforming the other baseline approaches significantly. 

\begin{figure*}[ht]
    \centering
    \subfigure[][\emph{MNIST}.]{
        \label{fig:9a}
        \includegraphics[width=0.6\columnwidth]{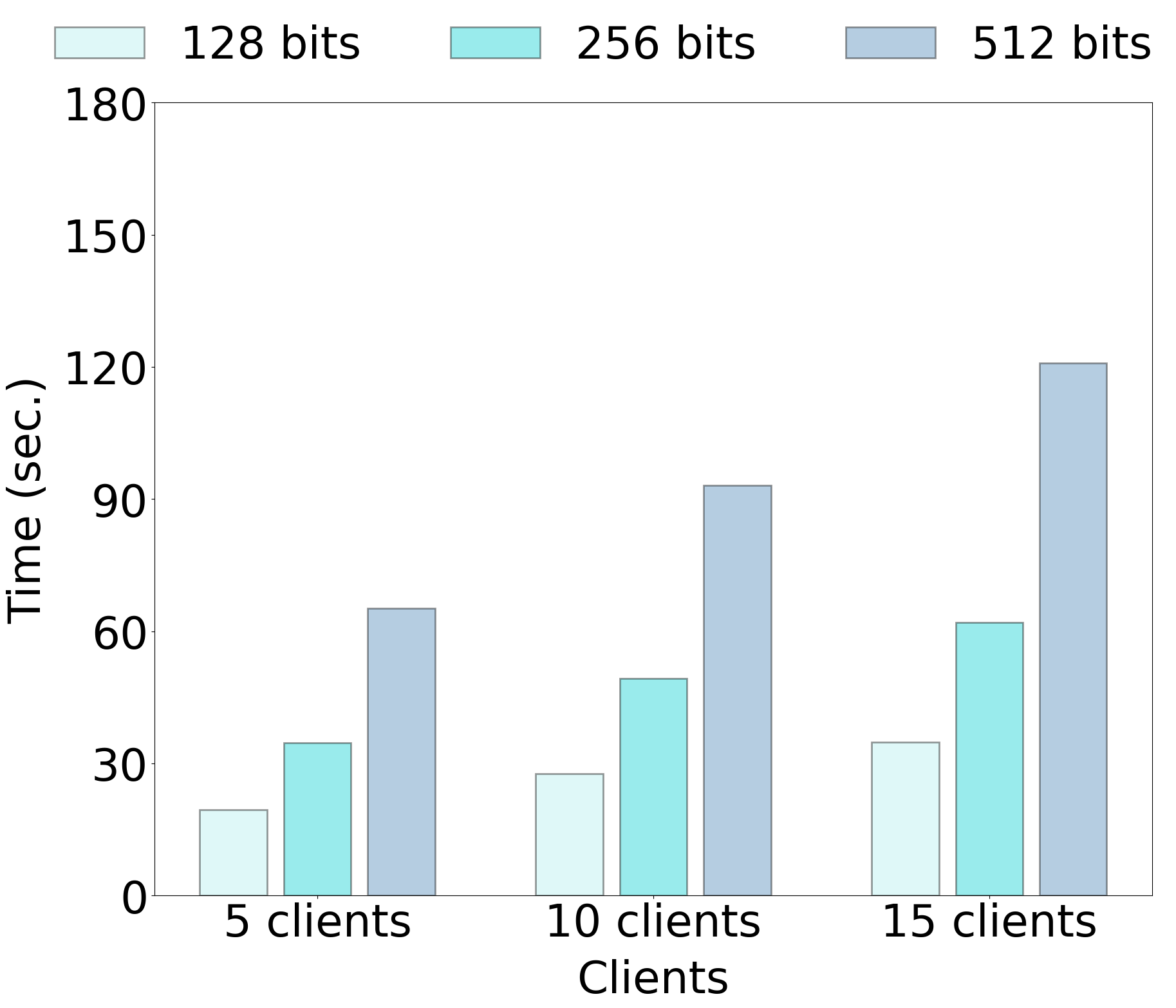}
    }
    \subfigure[][\emph{CIFAR-10}.]{
        \label{fig:9b}
        \includegraphics[width=0.6\columnwidth]{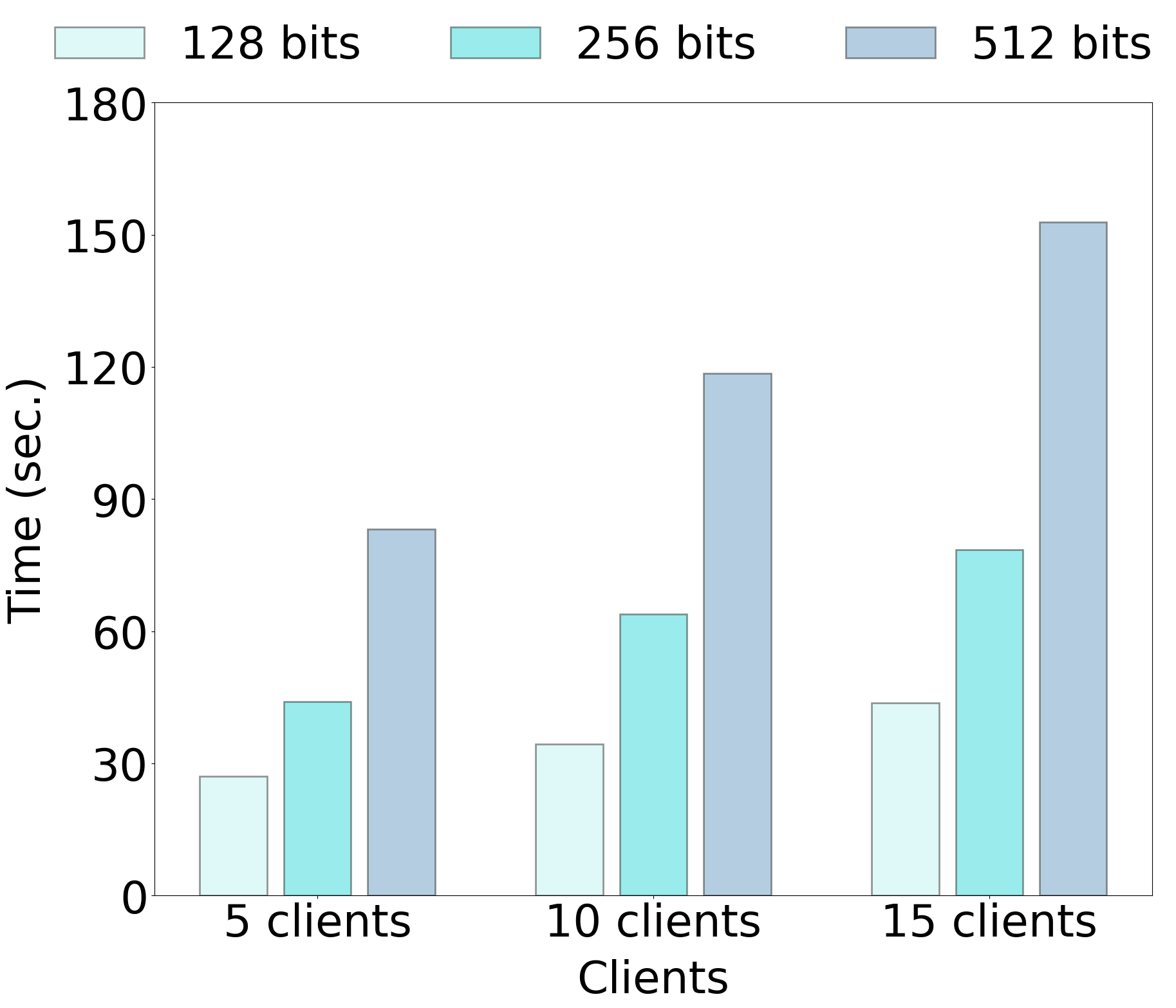}
    }
    \subfigure[][\emph{CIFAR-100}.]{
        \label{fig:9c}
        \includegraphics[width=0.6\columnwidth]{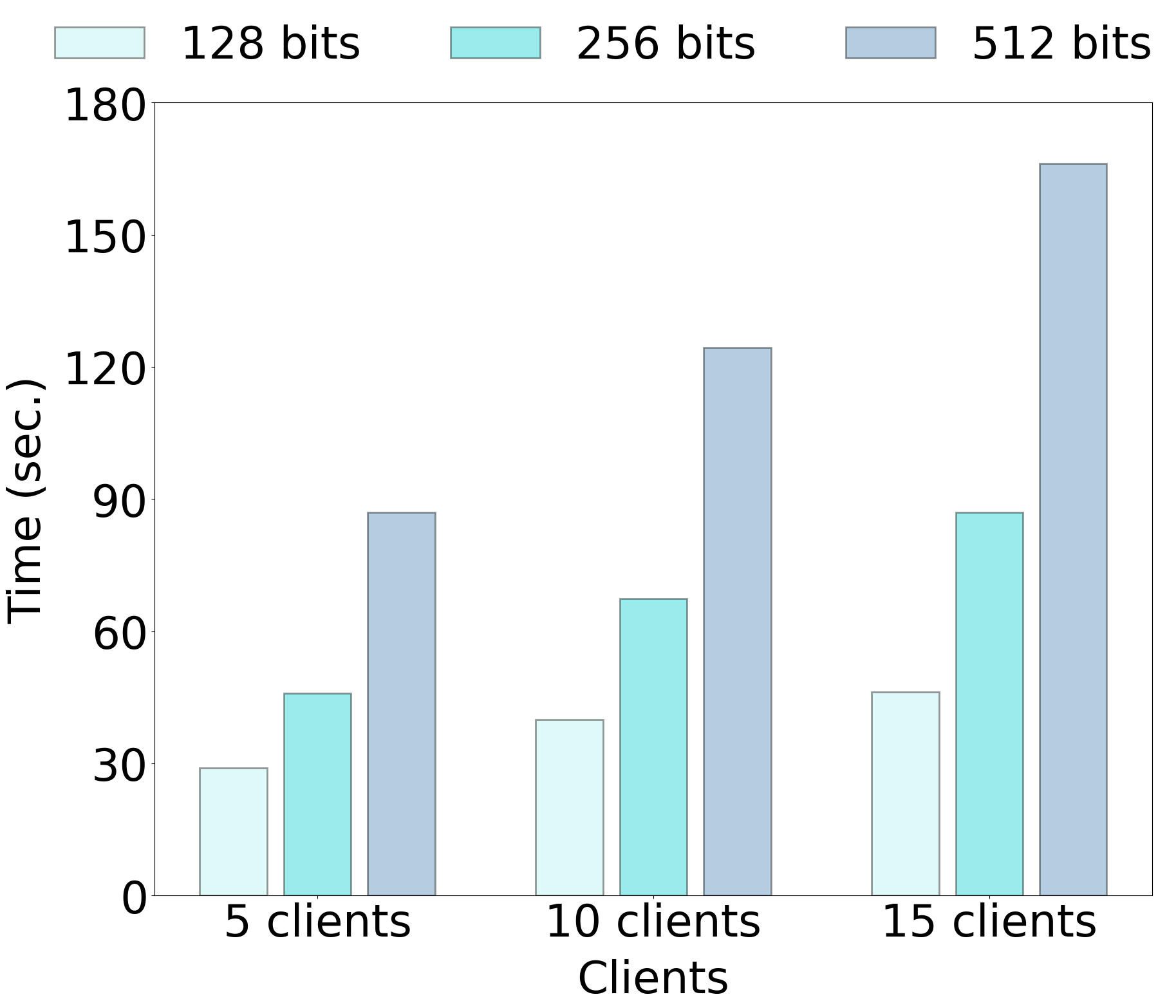}
    }
    \caption{Execution time for \emph{MNIST}, \emph{CIFAR-10}, and \emph{CIFAR-100}.}
    \label{fig:threedataset_BitsTime}
\end{figure*}

\subsubsection{Efficiency}\label{c:efficiency}
We evaluate the time cost to train a global model for one round in the blockchain under different public key lengths of Paillier encryption and different numbers of clients in our proposed \texttt{AerisAI}. We measure the total time required for a smart contract to execute within one round. There are four stages for the smart contract: i) \emph{Upload.} Each client invokes a transaction to upload the encrypted perturbed gradients and encrypted noise, which are aggregated by the smart contract. ii) \emph{Update.} The global encrypted model and noise are updated with the aggregated gradients and noise, respectively. iii) \emph{Model download.} The client invokes a transaction to query the encrypted global model. iv) \emph{Noise download.} The client invokes a transaction to query the encrypted global noise. The oracle decrypts the encrypted global noise first. Then, it encrypts the global noise and session key with AES and the corresponding attributes, respectively. The details are described in Sec. \ref{subc: workflow}, i.e., from S3 to S6.

Figs. \ref{fig:threedataset_BitsTime} show the result for MNIST, CIFAR-10, and CIFAR-100. The time cost increases as the length of the public key increases. This is because the time cost of aggregation and global model update depends on the length of the ciphertext. Although the execution grows as the length increases, it enhances the security of the system. Also, more clients participating in the system increase the time. However, thanks to CP-ABE, the oracle broadcasts the encrypted noise to multiple clients, and thus the time cost does not grow as much as the number of clients increases. The training time is still quite reasonable for our proposed \texttt{AriesAI}. 

\begin{figure}[ht]
    \subfigure[][5 clients]{
        \label{fig:6a}
        \includegraphics[width=0.45\columnwidth]{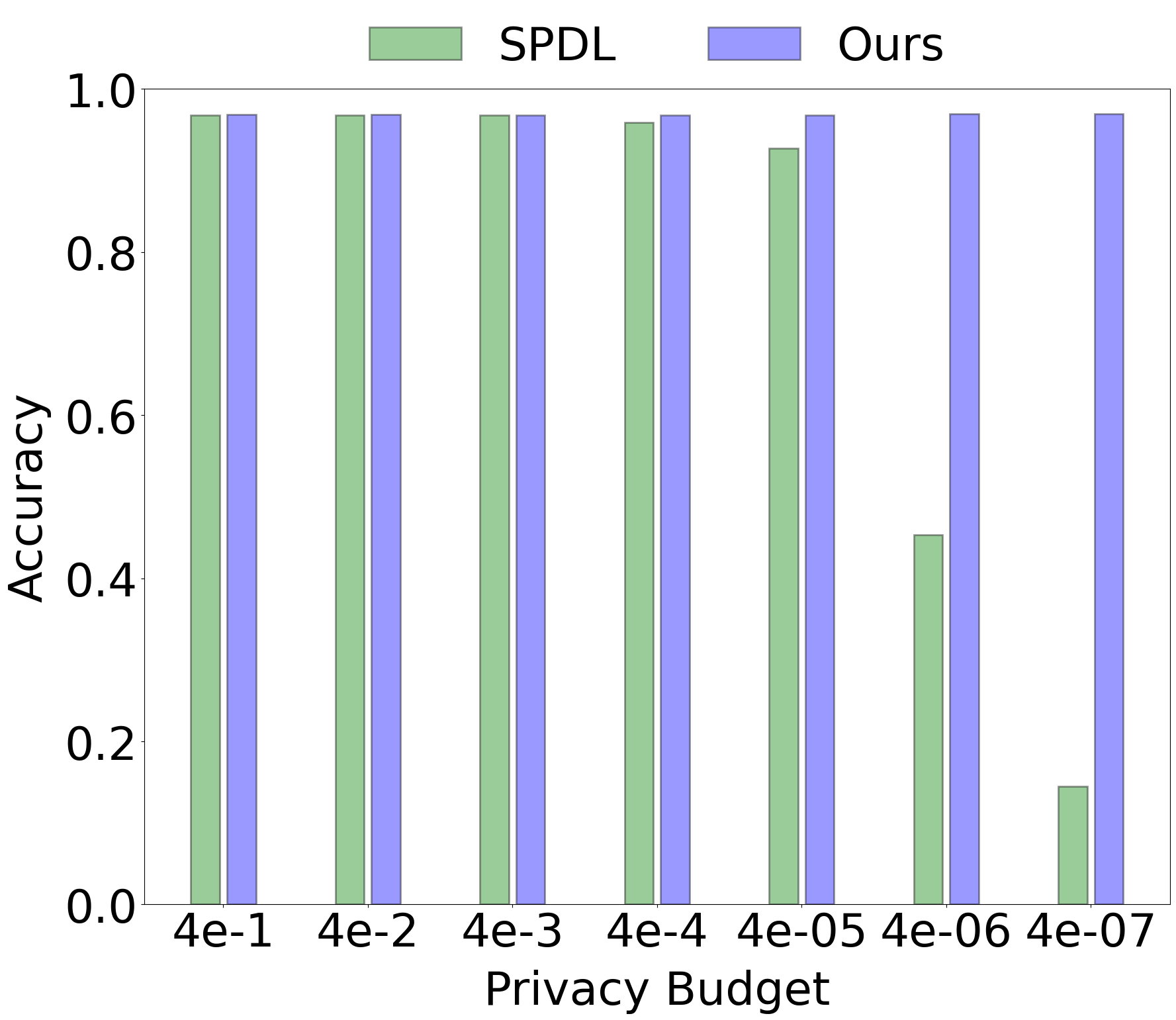}
    }
    \subfigure[][15 clients]{
        \label{fig:6c}
        \includegraphics[width=0.45\columnwidth]{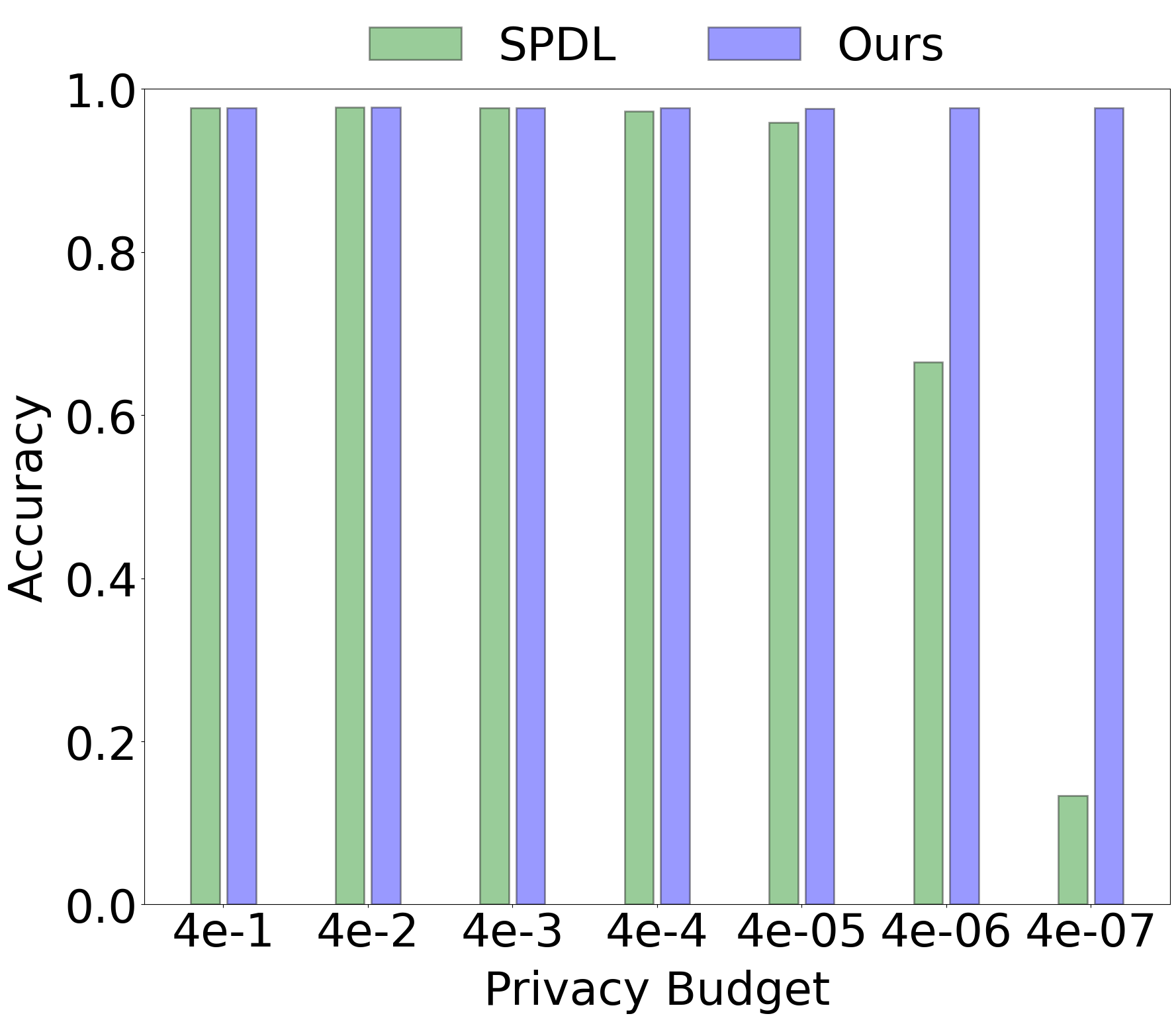}
    }
    \caption{Accuracy with different privacy budgets for \emph{MNIST} dataset.}
    \label{fig:MNIST_BudgetAccu}
\end{figure}

\begin{figure}[ht]
    \subfigure[][5 clients]{
        \label{fig:7a}
        \includegraphics[width=0.45\columnwidth]{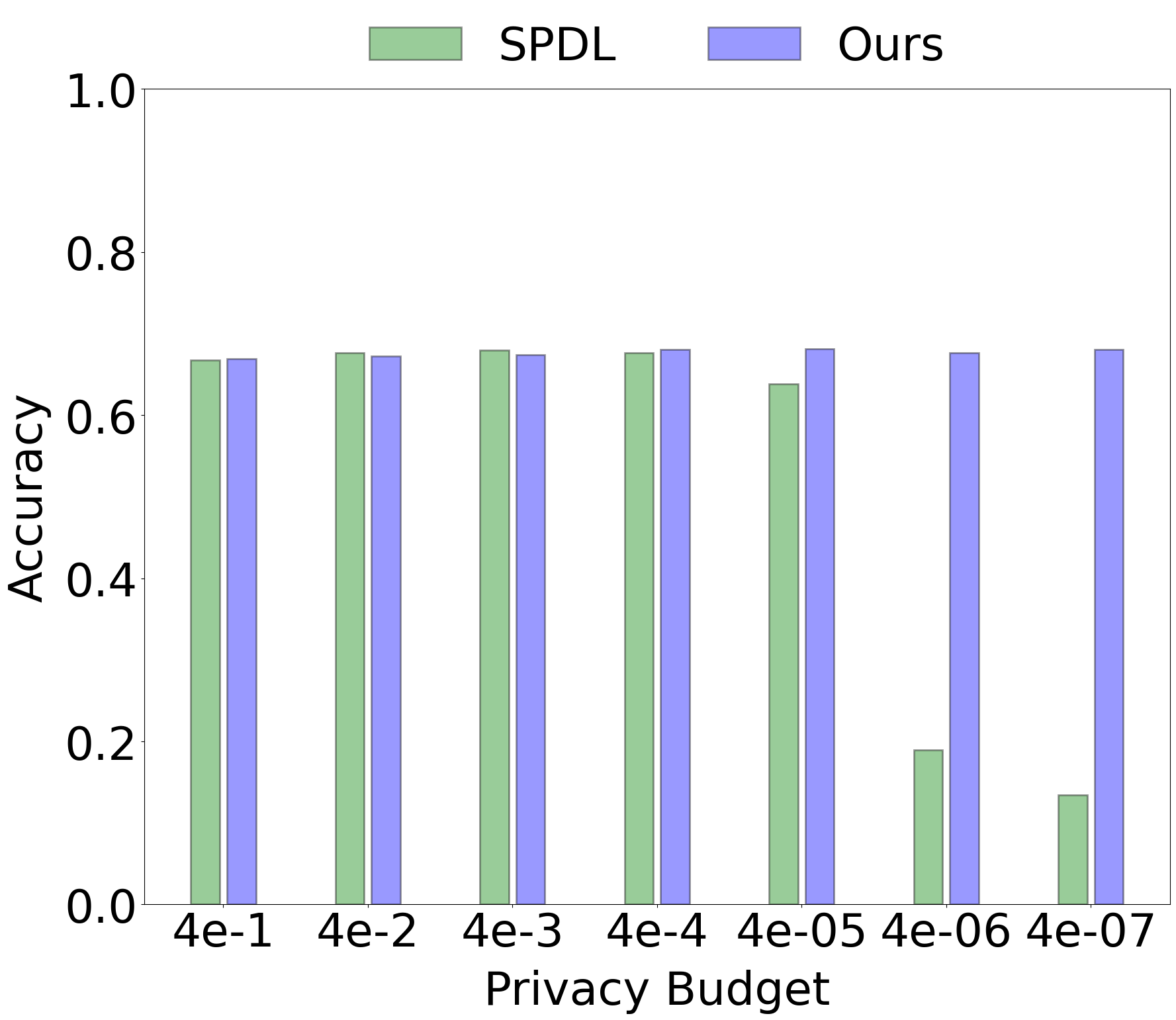}
    }
    \subfigure[][15 clients]{
        \label{fig:7c}
        \includegraphics[width=0.45\columnwidth]{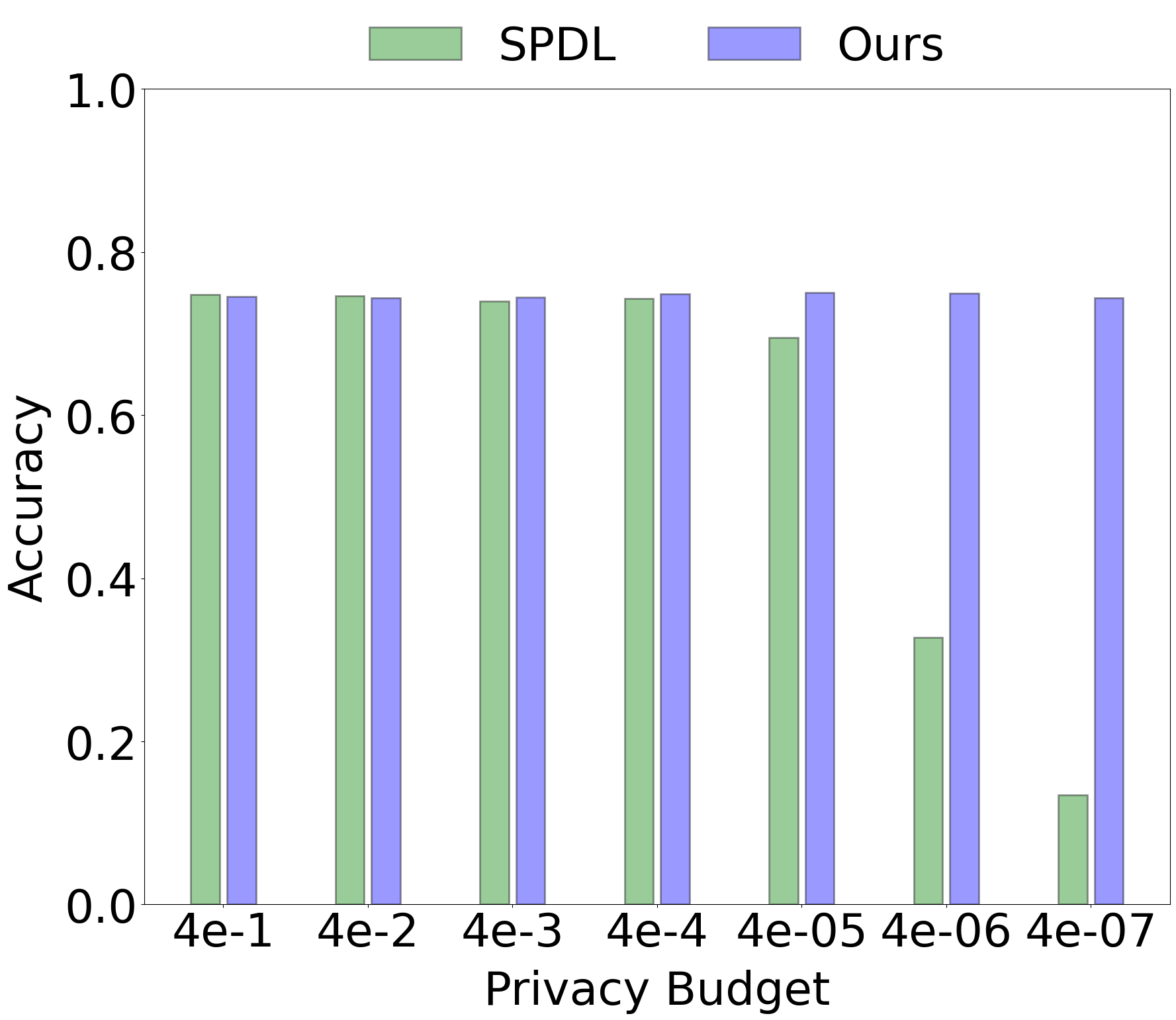}
    }
    \caption{Accuracy with different privacy budgets for \emph{CIFAR-10} dataset.}
    \label{fig:CIFAR10_BudgetAccu}
\end{figure}

\begin{figure}[!ht]
    \subfigure[][5 clients]{
        \label{fig:8a}
        \includegraphics[width=0.45\columnwidth]{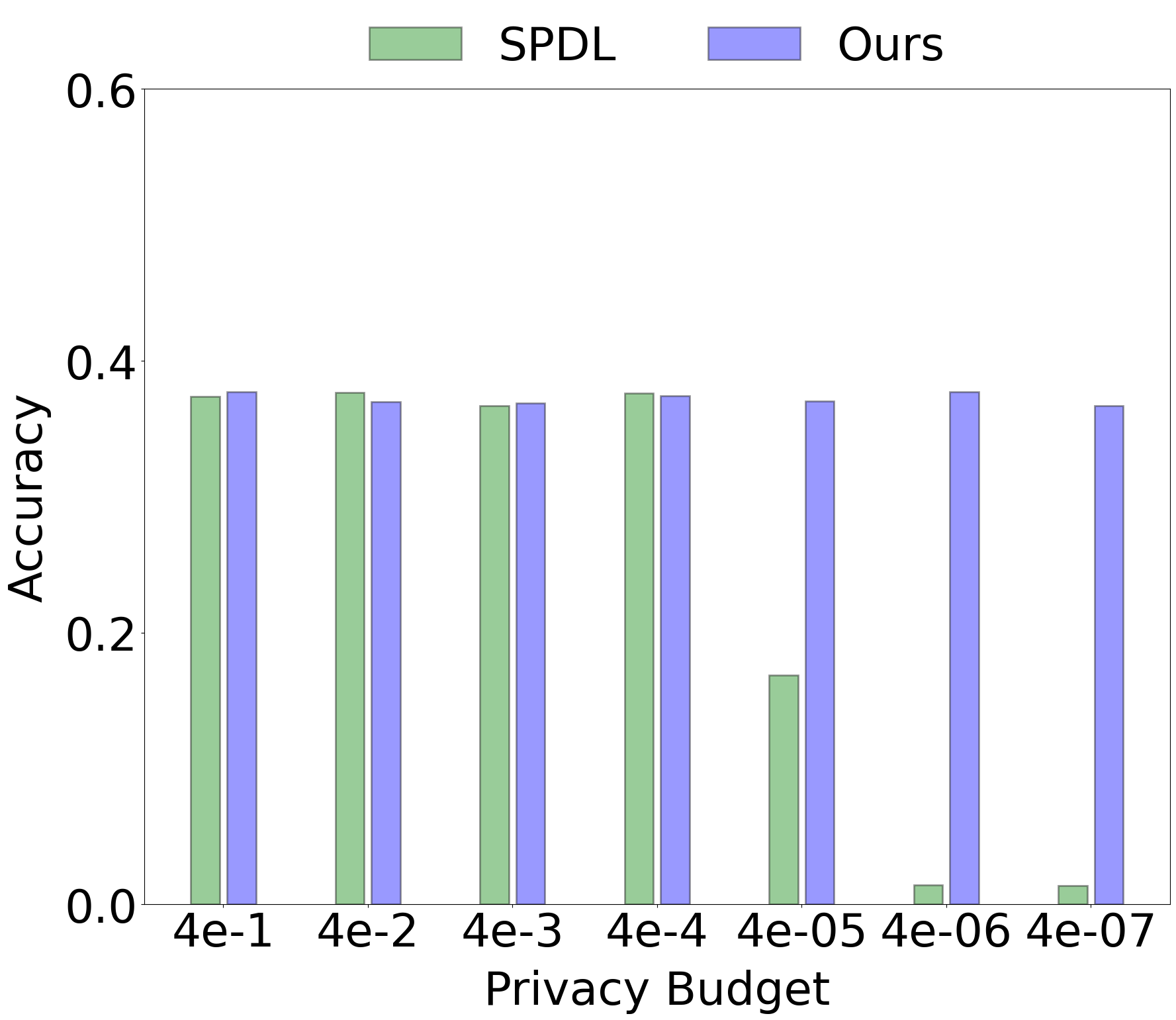}
    }
    \subfigure[][15 clients]{
        \label{fig:8c}
        \includegraphics[width=0.45\columnwidth]{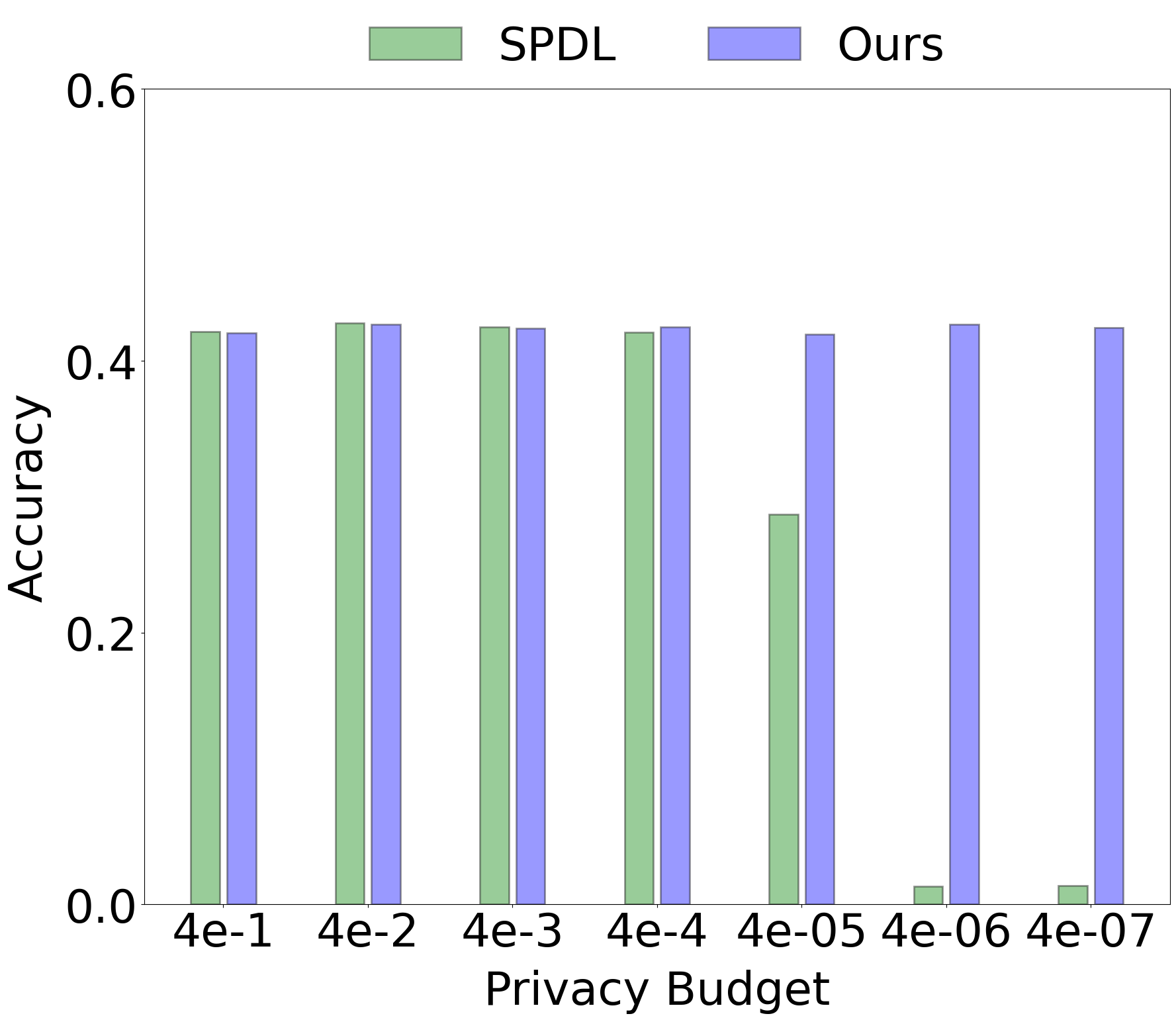}
    }
    \caption{Accuracy with different privacy budgets for \emph{CIFAR-100} dataset.}
    \label{fig:CIFAR100_BudgetAccu}
\end{figure}

\subsubsection{Impact of Privacy Budgets}
We further analyze the impact of different privacy budgets on the global model. For gradient perturbation, the noise is randomly chosen from the Gaussian distribution and injected into gradients. The privacy budget controls the standard deviation of the Gaussian distribution. A smaller privacy budget corresponds to a larger standard deviation, leading to a more dispersed distribution. In other words, with a smaller privacy budget, the noise is more likely to be larger values, which perturbs the gradients more significantly. Therefore, a smaller privacy budget represents a stronger privacy preservation capability but may degrade the performance of the global model more. 

In this set of experiments, we evaluate the global model after a 50-round training with different privacy budgets for MNIST, CIFAR-10, and CIFAR-100 datasets. The results are shown in Figs. \ref{fig:MNIST_BudgetAccu} (\emph{MNIST}), \ref{fig:CIFAR10_BudgetAccu} (\emph{CIFAR-10}), and \ref{fig:CIFAR100_BudgetAccu} (\emph{CIFAR-100}). We set the privacy budgets as $\{0.4, 0.04, ..., 0.0000004\}$, according to \cite{xu2022spdl} 
As shown in Figs. \ref{fig:MNIST_BudgetAccu}, \ref{fig:CIFAR10_BudgetAccu}, and \ref{fig:CIFAR100_BudgetAccu}, our proposed \texttt{AerisAI} (\texttt{ours}) achieves very good accuracy for various numbers of clients and different privacy budgets among the three datasets. In contrast, the performance of \texttt{SPDL} becomes much worse when the privacy budget is smaller. This indicates that our proposed \texttt{AerisAI} is able to achieve very high accuracy while preserving the privacy, outperforming the baseline \texttt{SPDL}.

\section{Conclusion}
\label{c:Conclusions}
In this paper, we propose AerisAI that builds a decentralized collaborative AI platform to enhance privacy without degrading the model performance. We provide formal proof of the security and functionality comparisons with the other state-of-the-art FL approaches. We also conduct extensive experiments on a consortium blockchain network. The results indicates that our proposed approach outperforms the other baselines significantly. 

\end{CJK}

\bibliographystyle{IEEEtran}
\bibliography{ref}

\end{document}